\documentclass[sigplan,nonacm]{acmart}

\usepackage{xspace}
\usepackage{tikz}
\usepackage{booktabs}
\usetikzlibrary{positioning, arrows.meta, shapes.geometric, calc, fit, backgrounds, decorations.pathreplacing, patterns}

\usepackage{listings}
\newcommand{\sys}{Vosti\xspace}

\newcommand{\pass}[1]{\textcolor{green!50!black}{#1}}
\newcommand{\fail}[1]{\textcolor{red}{#1}}

\lstdefinelanguage{Dafny}{
  morekeywords={datatype, method, function, returns, requires, ensures,
                modifies, var, if, else, while, forall, exists,
                seq, set, map, nat, int, bool, true, false},
  morecomment=[l]{//},
  morestring=[b]",
}
\lstdefinelanguage{Verus}{
  morekeywords={pub, struct, enum, fn, impl, trait, spec, proof, exec,
                open, closed, requires, ensures, invariant, decreases,
                forall, exists, let, ghost, tracked, mut, match, if, else,
                while, for, return, usize, u64, u32, bool, nat, int,
                Vec, Seq, Map, Set, Option, true, false},
  morecomment=[l]{//},
  morestring=[b]",
}
\definecolor{annText}{HTML}{24292E}
\definecolor{annKeyword}{HTML}{005CC5}
\definecolor{annDecorator}{HTML}{6F42C1}
\definecolor{annFunction}{HTML}{795E26}
\definecolor{annTensor}{HTML}{006B70}
\definecolor{annNumber}{HTML}{953800}
\definecolor{annComment}{HTML}{6A737D}
\definecolor{annBackground}{HTML}{F6F8FA}
\definecolor{annBorder}{HTML}{D1D5DA}
\lstdefinelanguage{VerifAnnotation}{
  sensitive=true,
  alsoletter={@},
  morekeywords=[1]{@verif,@params,@grid},
  morekeywords=[2]{same,pre,post,singleton},
  morekeywords=[3]{left,right},
  morekeywords=[4]{a,b,c}
}
\lstdefinestyle{verifannotation}{
  language=VerifAnnotation,
  basicstyle=\ttfamily\small\color{annText},
  keywordstyle=[1]\color{annDecorator},
  keywordstyle=[2]\color{annKeyword}\bfseries,
  keywordstyle=[3]\color{annFunction},
  keywordstyle=[4]\color{annTensor},
  backgroundcolor=\color{annBackground},
  frame=single,
  rulecolor=\color{annBorder},
  framerule=0.4pt,
  framesep=4pt,
  xleftmargin=5pt,
  xrightmargin=5pt,
  literate=
    {\#}{{\textcolor{annComment}{\#}}}1
    {0}{{\textcolor{annNumber}{0}}}1
    {1}{{\textcolor{annNumber}{1}}}1
    {2}{{\textcolor{annNumber}{2}}}1
    {3}{{\textcolor{annNumber}{3}}}1
    {4}{{\textcolor{annNumber}{4}}}1
    {5}{{\textcolor{annNumber}{5}}}1
    {6}{{\textcolor{annNumber}{6}}}1
    {7}{{\textcolor{annNumber}{7}}}1
    {8}{{\textcolor{annNumber}{8}}}1
    {9}{{\textcolor{annNumber}{9}}}1
}

\def\Snospace~{\S{}}

\newcommand{\paraspace}{\vspace{0.05in}}
\newcommand{\parab}[1]{\paraspace\noindent{\bf #1} }

\DeclareMathSizes{8.05}{8.05}{8.05}{8.05}
\DeclareMathSizes{10.01}{10.01}{8.05}{8.05}

\begin{document}
\title{\sys: Specifying, Implementing, and Verifying Deterministic LLM Inference}

\author{Jianxing Qin}
\affiliation{
  \institution{Duke University}
  \city{Durham}
  \state{NC}
  \country{United States}
}

\author{Alexander Du}
\affiliation{
  \institution{Duke University}
  \city{Durham}
  \state{NC}
  \country{United States}
}

\author{Danfeng Zhang}
\affiliation{
  \institution{Duke University}
  \city{Durham}
  \state{NC}
  \country{United States}
}

\author{Matthew Lentz}
\affiliation{
  \institution{Duke University}
  \city{Durham}
  \state{NC}
  \country{United States}
}

\author{Danyang Zhuo}
\affiliation{
  \institution{Duke University}
  \city{Durham}
  \state{NC}
  \country{United States}
}

\begin{abstract}

LLM inference systems may vary batch composition, prompt chunking, prefill/decode
execution, and KV-cache reuse, eviction, or recomputation. These
optimizations should not affect system outputs. Production systems,
including vLLM's batch-invariant mode and SGLang's deterministic mode, target
this goal but lack a formal system-level specification.

We formalize deterministic LLM inference: under a fixed model and deployment
configuration, requests with the same prompt and initial sampler state produce
bitwise-identical logits at corresponding output positions across executions.
Our tests find that these production modes produce different logits under some
execution variations. To address this limitation, we present \sys, an inference engine designed and
verified against this specification. \sys chooses kernels independently of runtime engine state and
ties cached KV values to their logical token prefixes. Its proof decomposes
at the engine/GPU kernel boundary: a Verus inductive proof establishes that
scheduling and the paged, prefix-sharing KV-cache preserve output logits,
while a Triton analyzer proves bitwise-equal selected kernel outputs across
batches, query lengths, and paged KV-cache layouts. \sys produces bitwise-identical
logits across every tested execution variation and achieves performance comparable to vLLM's
batch-invariant mode on decode-heavy workloads, while providing a stronger,
formally verified determinism guarantee.

\end{abstract}

\maketitle 
 

\section{Introduction}
LLM inference frameworks rarely execute a request in one fixed way.
As load and GPU memory pressure change, the server may co-batch the request with different companions~\cite{yu2022orca}, split its prompt into prefill chunks of different sizes~\cite{agrawal2024chunkedprefill}, compute a multi-token prefill or a single-token decode through speculative decoding~\cite{leviathan2023specdecoding, chen2023specdecoding}, and evict, relocate, share, or recompute its KV-cache state~\cite{kwon2023pagedattention}.
These choices may affect performance but should not affect model results: for a fixed model, prompt, and sampler state, the logits at a logical token position should not depend on how the server computes it.
Yet they can produce different logits.
GPU kernels select tilings and reduction orders based on execution shapes, and floating-point results depend on the resulting operation order~\cite{he2025defeating}.
The scheduler and paged cache also determine the tokens and KV state supplied to those kernels, so different executions can change kernel inputs.

Recent work addresses one source of these differences: batch-dependent kernels.
In under a year, batch-invariant kernels went from a blog post~\cite{he2025defeating} to production modes: vLLM ships a batch-invariant mode~\cite{vllm2025bitwise, vllm2026batchinvariance} and SGLang ships a deterministic mode~\cite{sglang2025deterministic}; DeepSeek-V4 builds end-to-end batch-invariant kernels into its training stack~\cite{deepseek2026v4}.
But batch composition is only one way a server varies execution, and no formal definition states what deterministic inference requires from the system as a whole.

\parab{Specification.}
Our first contribution is a system-level definition of deterministic LLM inference.
An LLM inference system is \emph{deterministic} if, for a fixed model, deployment configuration, prompt, and sampler state, every observed execution produces bitwise-identical logits at each output position; equal sampler states then yield prefix-comparable output traces.
The definition constrains observable outputs only; it does not assume batching, prompt chunking, or a KV-cache.
We vary batch sizes, chunked-prefill budgets, prefill versus decode, and prefix reuse in production LLM inference systems.
Both vLLM’s batch-invariant mode and SGLang’s deterministic mode produce logit mismatches.
Representative mismatches arise from shape-dependent kernel specialization and different attention reduction arithmetic across execution paths.
These results expose scope limits rather than necessarily indicating bugs: one kernel configuration's batch invariance does not establish determinism across chunking, prefill/decode transitions, or cache reuse.

\parab{\sys.} Our second contribution is the design and implementation of \sys, an LLM inference system designed and verified to provide deterministic inference. \sys's engine is implemented and formally verified using Verus~\cite{lattuada2024verus}.
It continuously batches prefill chunks and decode tokens, where each may carry arbitrary KV-cache pages.

We prove in Verus that each engine step produces logits determined solely by the request's token history. The proof maintains a central KV-cache invariant alongside request, layout, and
write-isolation invariants. The KV-cache invariant ensures that every live or
retained reusable KV entry equals the value obtained by recomputing its token
prefix without cache.
Induction over engine steps then establishes request determinism.

\sys preserves engine optimizations such as continuous batching, chunked
prefill, and KV-cache reuse, while constraining kernel choices that could
change floating-point results. It uses fixed, runtime-independent kernel
configurations, retaining tiling and independent parallelism.
Dispatch among bitwise-equivalent kernels remains a future extension.
Our analyzer translates annotated Triton~\cite{tillet2019triton} into a tile-level
IR and uses Z3~\cite{demoura2008z3} to prove bitwise equality of selected
kernel outputs across kernel invocations.
At runtime, \sys rejects a missing configuration rather than autotuning.

\parab{Evaluation.}
\sys passes all 5,488 bitwise comparisons across seven Llama and Gemma
models, covering batching, chunked prefill, prefill--decode equivalence, and prefix
reuse. Its performance is comparable to production inference engines on
decode-heavy workloads. On an H200 GPU, \sys achieves
$1.36$--$3.19\times$ the decode throughput of vLLM's batch-invariant modes
and completes sessions $1.14$--$2.20\times$ faster in a synthetic agentic
workload, while remaining slower than both default modes and SGLang's
deterministic modes.

\section{Specifying Deterministic LLM Inference}
\label{sec:background}

\begin{table}[t]
\centering
\footnotesize
\setlength{\tabcolsep}{3pt}
\caption{Motivating determinism tests of \sys and production inference engines.
Each cell reports bitwise-matched logit pairs divided by total compared pairs.
Green denotes full agreement; red denotes at least one mismatch.}
\label{tab:production-motivation}
\begin{tabular}{lllcccc}
\toprule
\textbf{Engine} &
\textbf{Attn.} &
\textbf{Model} &
\textbf{Batch} &
\textbf{Chunk} &
\textbf{P--D} &
\textbf{Prefix} \\
\midrule
vLLM        & FA3    & Llama3.1-8B & \fail{177/296} & \fail{49/90} & \fail{53/384}  & \fail{3/14}  \\
vLLM Inv.   & FA3    & Llama3.1-8B & \pass{296/296} & \pass{90/90} & \pass{384/384} & \pass{14/14} \\
vLLM Inv.   & Triton & Llama3.1-8B & \pass{296/296} & \pass{90/90} & \pass{384/384} & \pass{14/14} \\
SGLang      & FA3    & Llama3.1-8B & \fail{156/296} & \fail{50/90} & \fail{59/384}  & \fail{4/14}  \\
SGLang Det. & FA3    & Llama3.1-8B & \pass{296/296} & \pass{90/90} & \pass{384/384} & \pass{14/14} \\
SGLang Det. & Triton & Llama3.1-8B & \pass{296/296} & \pass{90/90} & \fail{3/384}   & \fail{13/14} \\
\sys        & Triton & Llama3.1-8B & \pass{296/296} & \pass{90/90} & \pass{384/384} & \pass{14/14} \\
\midrule
vLLM        & FA3    & Gemma3-4B & \fail{70/296}  & \fail{0/90}  & \fail{5/384}   & \fail{3/14}  \\
vLLM Inv.   & FA3    & Gemma3-4B & \fail{193/296} & \fail{0/90}  & \fail{5/384}   & \fail{11/14} \\
vLLM Inv.   & Triton & Gemma3-4B & \fail{242/296} & \fail{0/90}  & \fail{192/384} & \fail{13/14} \\
SGLang      & FA3    & Gemma3-4B & \fail{81/296}  & \fail{42/90} & \fail{5/384}   & \fail{2/14}  \\
SGLang Det. & FA3    & Gemma3-4B & \fail{206/296} & \fail{70/90} & \fail{5/384}   & \fail{10/14} \\
SGLang Det. & Triton & Gemma3-4B & \fail{256/296} & \fail{70/90} & \fail{3/384}   & \fail{9/14}  \\
\sys        & Triton & Gemma3-4B & \pass{296/296} & \pass{90/90} & \pass{384/384} & \pass{14/14} \\
\bottomrule
\end{tabular}
\end{table}

\subsection{Why Deterministic Inference?}
An LLM server can realize one output token sequence through many executions.  Continuous
batching changes a request's companions as requests arrive and finish; chunked
prefill partitions a prompt according to a token budget; speculative decoding
computes several candidate positions in an extend step; and a paged KV-cache may
share, relocate, evict, or recompute a logical prefix~\cite{yu2022orca,
kwon2023pagedattention, leviathan2023specdecoding, chen2023specdecoding}.
These choices should change performance, not model results.

Both GPU kernels and the inference engine can change these results.  GPU kernels
select tilings, reduction orders, and sometimes entire implementations based
on tensor shapes.  As floating-point addition is not associative, changing
the reduction order can change the resulting bit pattern~\cite{he2025defeating}.
The scheduler and cache determine the tokens, positions, cached prefix, and metadata supplied to those kernels.  Determinism must therefore hold
across the whole system: invariant kernels cannot compensate for an engine that
supplies a different input, and correct engine state is insufficient if
a kernel changes its value-relevant operation order.

Different logits can change stochastic draft
acceptance~\cite{leviathan2023specdecoding, chen2023specdecoding} and violate the
prefill--decode equivalence expected by speculative
verification~\cite{zhang2026batchspec}.  They can also create a mismatch between
rollout and training log-probabilities in reinforcement
learning~\cite{vllm2025bitwise}, make a
continuation depend on whether KV state was retained or recomputed, and prevent
an operator from reproducing one request independently of the load and cache
history under which it originally ran.

\subsection{Measuring Production Deterministic Modes}
\label{sec:production-measurements}

We measure whether production inference engines and \sys preserve bitwise
logit equality under different execution variations. We run vLLM 0.28.0
in its default and batch-invariant modes, and SGLang~\cite{zheng2024sglang}
0.5.19 in its default
and deterministic modes. For both production deterministic modes, we test
FlashAttention 3 (FA3)~\cite{shah2024flashattention3} and Triton attention backends. Default modes select
their backends automatically; both select FA3 for both models.
We also run \sys using its verified Triton kernels. All experiments use Llama3.1-8B and text-only Gemma3-4B
in BF16 on one NVIDIA H200. Each comparison checks
bitwise equality of the complete raw next-token logit vectors.
Table~\ref{tab:production-motivation} reports matched pairs over total
compared pairs.

\parab{Test construction.}
\emph{Batch composition} runs the same prompts individually and in
batches of different sizes and orders.
We compare each prompt's last logits.
\emph{Chunked prefill} prefills the same prompt under different chunk
budgets, including an unchunked reference, and compares the last logits.
\emph{Prefill--decode} generates a fixed-length response, then
prefills the original prompt plus successive prefixes of that response
in a separate engine. We compare the logits
at corresponding prediction positions, ensuring identical token histories.
\emph{Prefix reuse} compares cold-prefill logits against those obtained
after another request populates the cache in the same engine as the warm
request. Cases cover full and partial prefix reuse; recorded cache-hit counts validate reuse where required.

\parab{Results.}
\sys matches all pairs per model. On Llama3.1-8B, both vLLM batch-invariant backends and SGLang deterministic
FA3 match every pair. SGLang deterministic Triton matches all batch and
chunk comparisons, but only 3/384 prefill--decode pairs and 13/14 prefix
pairs.
On Gemma3-4B, all four production deterministic configurations exhibit
mismatches in every category.
These modes usually reduce mismatches relative to defaults, but bitwise logit agreement still depends on model, backend, and execution variation.

These tests are not exhaustive verification:
passing does not establish determinism beyond the tested cases.
Conversely, a mismatch demonstrates numerical variation, not necessarily
a bug for the intended use case.

\subsection{A System-Level Formulation}
\label{sec:serving-spec}

Fix an inference configuration \(W\), comprising the model weights, numerical
dtypes, and all value-relevant aspects of the platform and implementation.
An output record is a pair
\[
  (\ell,y)
  \in \mathsf{Logits} \times \mathsf{Token},
\]
where \(\ell\) is the logit bit pattern and \(y\) is the emitted token.

Let \(\mathsf{Id}\) be the set of request identifiers.  At each engine step,
the observable output is a finite partial map
\[
  O_i:\mathsf{Id}\rightharpoonup
      (\mathsf{Logits}\times\mathsf{Token}).
\]
Thus, an engine step may emit one output record for each of zero or more
requests, without imposing an order among requests processed in the same
step.

For a request \(r\), let \(s_0\) be any engine state reachable under \(W\)
immediately after \(r\) is admitted and before it produces any output.
This state may contain other requests and retained cache state.
A finite execution starting from \(s_0\) has the form
\[
  s_0 \xrightarrow{O_1} s_1
      \xrightarrow{O_2} \cdots
      \xrightarrow{O_n} s_n .
\]
We identify its observable trace with the finite sequence
\[
  T=(O_1,O_2,\ldots,O_n).
\]
Define the per-request view of \(r\) by
\[
  \mathsf{view}_r(T)
  =
  \bigl(
    O_i(r.\mathsf{id})
    \mathrel{\big|}
    1\leq i\leq n
    \ \land\
    r.\mathsf{id}\in\operatorname{dom}(O_i)
  \bigr),
\]
where the records are ordered by increasing \(i\).  This view retains the
observable history of \(r\), but does not retain its identifier, its batch
position, steps in which it produced no output, or outputs belonging to other
requests.

\begin{definition}[Deterministic LLM inference]
\label{def:system-determinism}
An inference system is \emph{deterministic} if, for every pair of finite
executions \(T_1,T_2\) that share configuration \(W\), and every
pair of requests \(r_1,r_2\) newly admitted in their respective initial
states as above, whenever their prompts and initial sampler states are the same:
\[
  r_1.\mathsf{prompt}=r_2.\mathsf{prompt}
  \quad\text{and}\quad
  r_1.\mathsf{sampler}_0=r_2.\mathsf{sampler}_0,
\]
their per-request views are prefix-comparable:
\[
  \mathsf{view}_{r_1}(T_1)
  \preceq
  \mathsf{view}_{r_2}(T_2)
  \quad\lor\quad
  \mathsf{view}_{r_2}(T_2)
  \preceq
  \mathsf{view}_{r_1}(T_1).
\]
Here, \(\preceq\) denotes the prefix relation. Equality of logit components means equality of their bit patterns.
\end{definition}

Intuitively, a global execution may interleave many requests and may batch
them differently at each engine step.  The operation
\(\mathsf{view}_r\) removes these scheduling details and retains only the
observable history of request \(r\).  It also erases the request identifier:
identifiers distinguish request instances but must not affect corresponding outputs.

Because the executions are finite and may schedule requests differently, their
per-request views need not have equal lengths.  Prefix comparability requires
the shorter view to agree exactly with the beginning of the longer one.
Equivalently, corresponding requests produce identical logit bits and tokens
at every output index reached by both executions.

This definition is not a liveness property.  It constrains
outputs whenever both executions produce the corresponding records, but does
not require an admitted request to be scheduled, to make progress, or to
complete.
It is also not a functional-correctness specification: a system may satisfy determinism while producing logits or tokens that differ from those prescribed by the intended model semantics.

\section{Overview}
\label{sec:overview}

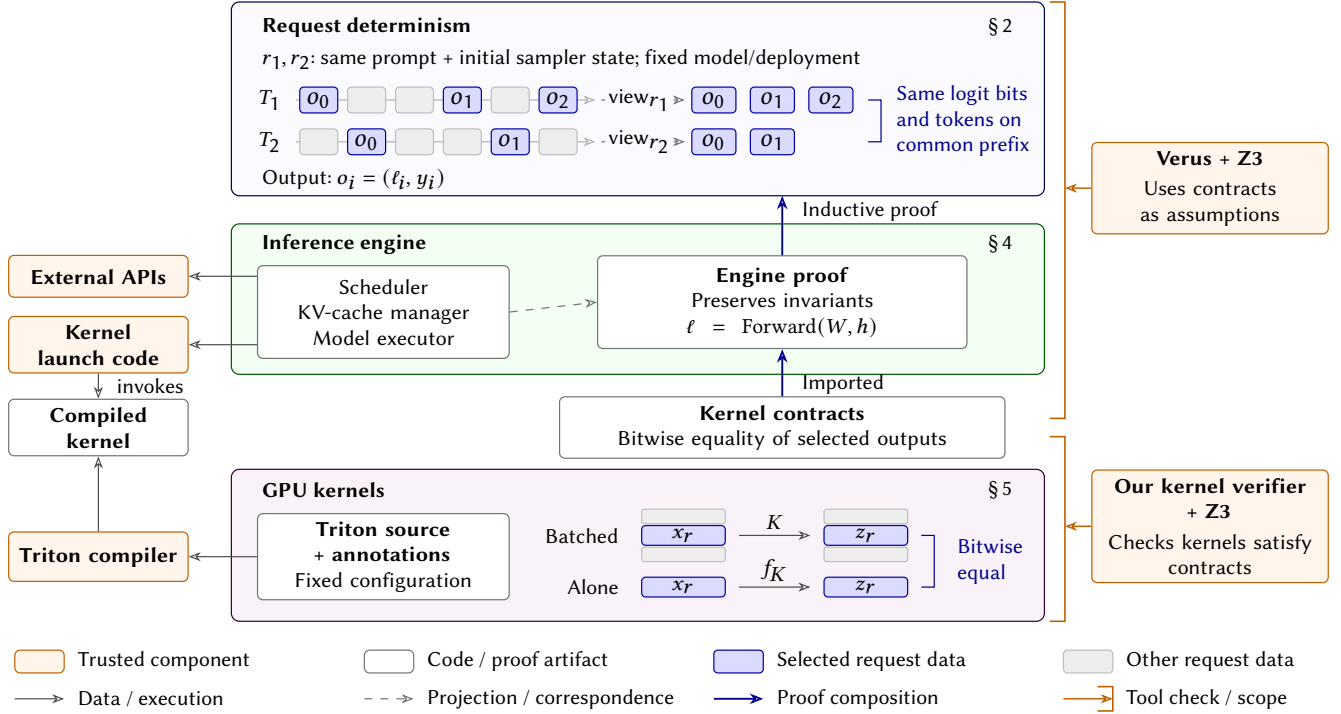
\begin{figure*}[t]
  \centering
  \resizebox{\textwidth}{!}{
\begin{tikzpicture}[
  x=1cm, y=1cm,
  font=\sffamily\fontsize{8.05}{9.5}\selectfont,
  line width=0.45pt,
  box/.style={draw=black!55, fill=white, rounded corners=2pt,
    align=center, inner sep=3pt},
  heading/.style={font=\sffamily\bfseries\fontsize{8.05}{9.5}\selectfont},
  smalltext/.style={font=\sffamily\fontsize{8.05}{9.5}\selectfont},
  outputrecord/.style={font=\sffamily\fontsize{10.01}{11}\selectfont},
  legendtext/.style={smalltext, anchor=west, inner sep=0pt},
  tool/.style={box, draw=orange!75!black, fill=orange!8},
  requestband/.style={draw=blue!30!black, fill=blue!2, rounded corners=3pt},
  engineband/.style={draw=green!35!black, fill=green!5, rounded corners=3pt},
  kernelband/.style={draw=violet!40!black, fill=violet!5, rounded corners=3pt},
  arrow/.style={-{Stealth[open,length=1.6mm,width=1.2mm]}, draw=black!70},
  checkarrow/.style={arrow, draw=orange!75!black, line width=0.55pt},
  checkscope/.style={draw=orange!75!black, line width=0.6pt},
  proofarrow/.style={-{Stealth[open,length=1.8mm,width=1.4mm]},
    draw=blue!55!black, line width=0.7pt},
  mapping/.style={arrow, dashed, draw=black!55},
  annotation/.style={smalltext, fill=none, inner sep=1pt, align=center},
  selected/.style={draw=blue!65!black, fill=blue!14, rounded corners=1.5pt,
    minimum width=0.58cm, minimum height=0.34cm, inner sep=1pt},
  other/.style={draw=black!25, fill=black!7, rounded corners=1.5pt,
    minimum width=0.50cm, minimum height=0.34cm, inner sep=0pt}
]
  \useasboundingbox (0,-0.82) rectangle (17.6,8.60);

  \draw[requestband] (3.15,6.05) rectangle (13.85,8.60);
  \draw[engineband] (3.15,3.70) rectangle (13.85,5.68);
  \draw[kernelband] (3.15,0.45) rectangle (13.85,2.45);

  \node[heading, anchor=west] at (3.43,8.27) {Request determinism};
  \node[smalltext, anchor=east] at (13.55,8.27) {\S\,2};
  \node[smalltext, anchor=west] at (3.43,7.85)
    {$r_{\mathsf{1}},r_{\mathsf{2}}$: same prompt + initial sampler state; fixed model/deployment};
  \node[anchor=east] at (3.92,7.31) {$T_{\mathsf{1}}$};
  \node[anchor=east] at (3.92,6.76) {$T_{\mathsf{2}}$};
  \draw[arrow, draw=black!30] (4.00,7.31) -- (7.94,7.31);
  \draw[arrow, draw=black!30] (4.00,6.76) -- (7.94,6.76);
  \foreach \x/\label in {4.30/0,6.19/1,7.45/2}
    \node[selected, outputrecord, minimum width=0.50cm] at (\x,7.31) {$o_{\mathsf{\label}}$};
  \foreach \x/\label in {4.93/0,6.82/1}
    \node[selected, outputrecord, minimum width=0.50cm] at (\x,6.76) {$o_{\mathsf{\label}}$};
  \foreach \x in {4.93,5.56,6.82}
    \node[other] at (\x,7.31) {};
  \foreach \x in {4.30,5.56,6.19,7.45}
    \node[other] at (\x,6.76) {};
  \node[annotation] (viewone) at (8.52,7.31) {$\mathsf{view}_{r_{\mathsf{1}}}$};
  \node[annotation] (viewtwo) at (8.52,6.76) {$\mathsf{view}_{r_{\mathsf{2}}}$};
  \draw[mapping, -] (8.02,7.31) -- (viewone.west);
  \draw[mapping] (viewone.east) -- (9.12,7.31);
  \draw[mapping, -] (8.02,6.76) -- (viewtwo.west);
  \draw[mapping] (viewtwo.east) -- (9.12,6.76);
  \foreach \y in {7.31,6.76} {
    \foreach \x/\label in {9.50/0,10.27/1}
      \node[selected, outputrecord] at (\x,\y) {$o_{\mathsf{\label}}$};
  }
  \node[selected, outputrecord] at (11.04,7.31) {$o_{\mathsf{2}}$};
  \draw[draw=blue!55!black] (11.53,6.76) -- (11.70,6.76)
    -- (11.70,7.31) -- (11.53,7.31);
  \node[smalltext, align=left, anchor=west, text=blue!55!black]
    at (11.78,7.035) {Same logit bits\\and tokens on\\common prefix};
  \node[smalltext, anchor=west] at (3.43,6.27)
    {Output: $o_i=(\ell_i,y_i)$};

  \node[heading, anchor=west] at (3.43,5.42) {Inference engine};
  \node[smalltext, anchor=east] at (13.55,5.42) {\S\,4};
  \node[box, text width=3.10cm, minimum height=1.22cm] (engine) at (5.15,4.52)
    {Scheduler\\KV-cache manager\\Model executor};
  \node[box, text width=4.65cm, minimum height=1.22cm] (engineproof) at (10.40,4.65)
    {\textbf{Engine proof}\\Preserves invariants\\
     $\ell=\operatorname{Forward}(W,h)$};
  \draw[mapping] (engine.east) -- (engineproof.west);
  \draw[proofarrow] (engineproof.north) -- (10.40,6.05);
  \node[annotation, anchor=west] at (10.62,5.85) {Inductive proof};

  \node[tool, text width=2.15cm, minimum height=0.55cm] (external) at (1.40,4.99)
    {\textbf{External APIs}};
  \node[tool, text width=2.15cm, minimum height=0.65cm] (launcher) at (1.40,4.09)
    {\textbf{Kernel launch code}};
  \node[box, text width=2.15cm, minimum height=0.65cm] (gpu) at (1.40,3.00)
    {\textbf{Compiled kernel}};
  \node[tool, text width=2.15cm, minimum height=0.65cm] (compiler) at (1.40,1.30)
    {\textbf{Triton compiler}};
  \draw[arrow] (engine.west |- external.east) -- (external.east);
  \draw[arrow] (engine.west |- launcher.east) -- (launcher.east);
  \draw[arrow] (launcher.south) --
    node[annotation, right, xshift=2mm] {invokes} (gpu.north);
  \draw[arrow] (compiler.north) -- (gpu.south);

  \node[box, text width=5.65cm, minimum height=0.78cm] (contract) at (10.40,3.00)
    {\textbf{Kernel contracts}\\Bitwise equality of selected outputs};
  \draw[proofarrow] (contract.north) -- (engineproof.south);
  \node[annotation, anchor=west] at (10.62,3.57) {Imported};

  \node[heading, anchor=west] at (3.43,2.18) {GPU kernels};
  \node[smalltext, anchor=east] at (13.55,2.18) {\S\,5};
  \node[box, text width=3.10cm, minimum height=1.03cm] (source) at (5.15,1.30)
    {\textbf{Triton source + annotations}\\Fixed configuration};
  \draw[arrow] (source.west) -- (compiler.east);

  \begin{scope}[yshift=-0.40cm]
  \node[smalltext, anchor=east] at (8.37,1.98) {Batched};
  \node[smalltext, anchor=east] at (8.37,1.30) {Alone};
  \foreach \x in {9.10,11.50} {
    \foreach \y in {1.73,1.98,2.23}
      \draw[other, rounded corners=1pt]
        (\x-0.55,\y-0.095) rectangle (\x+0.55,\y+0.095);
  }
  \foreach \y in {1.98,1.30} {
    \node[selected, minimum width=1.10cm, minimum height=0.26cm,
      smalltext] at (9.10,\y) {$x_r$};
    \node[selected, minimum width=1.10cm, minimum height=0.26cm,
      smalltext] at (11.50,\y) {$z_r$};
    \draw[arrow] (9.83,\y) -- (10.77,\y);
  }
  \node[annotation, above] at (10.30,2.02) {$K$};
  \node[annotation, above] at (10.30,1.34) {$f_K$};
  \draw[draw=blue!55!black] (12.23,1.30) -- (12.40,1.30)
    -- (12.40,1.98) -- (12.23,1.98);
  \node[align=left, anchor=west, text=blue!55!black]
    at (12.58,1.64) {Bitwise\\equal};
  \end{scope}

  \draw[checkscope] (13.92,8.60) -- (14.12,8.60)
    -- (14.12,3.12) -- (13.92,3.12);
  \node[tool, smalltext, text width=2.90cm, minimum height=1.15cm] (verus) at (16.03,6.15)
    {\textbf{Verus + Z3}\\[2pt]
     Uses contracts\\as assumptions};
  \draw[checkarrow] (verus.west) -- (14.12,6.15);
  \draw[checkscope] (13.92,2.88) -- (14.12,2.88)
    -- (14.12,0.45) -- (13.92,0.45);
  \node[tool, smalltext, text width=2.90cm, minimum height=1.45cm] (kernelverifier) at (16.03,1.70)
    {\textbf{Our kernel verifier\\+ Z3}\\[2pt]
     Checks kernels satisfy\\contracts};
  \draw[checkarrow] (kernelverifier.west) -- (14.12,1.70);

  \begin{scope}[yshift=0.45cm]
  \node[tool, minimum width=0.65cm, minimum height=0.30cm, inner sep=0pt]
    at (0.625,-0.53) {};
  \node[legendtext] at (1.12,-0.53) {Trusted component};
  \node[box, minimum width=0.65cm, minimum height=0.30cm, inner sep=0pt]
    at (5.225,-0.53) {};
  \node[legendtext] at (5.72,-0.53) {Code / proof artifact};
  \node[selected, minimum width=0.65cm, minimum height=0.30cm] at (9.825,-0.53) {};
  \node[legendtext] at (10.32,-0.53) {Selected request data};
  \node[other, minimum width=0.65cm, minimum height=0.30cm] at (14.425,-0.53) {};
  \node[legendtext] at (14.92,-0.53) {Other request data};
  \draw[arrow] (0.30,-1.02) -- (0.95,-1.02);
  \node[legendtext] at (1.12,-1.02) {Data / execution};
  \draw[mapping] (4.90,-1.02) -- (5.55,-1.02);
  \node[legendtext] at (5.72,-1.02) {Projection / correspondence};
  \draw[proofarrow] (9.50,-1.02) -- (10.15,-1.02);
  \node[legendtext] at (10.32,-1.02) {Proof composition};
  \draw[checkscope] (14.55,-0.86) -- (14.75,-0.86)
    -- (14.75,-1.18) -- (14.55,-1.18);
  \draw[checkarrow] (14.10,-1.02) -- (14.75,-1.02);
  \node[legendtext] at (14.92,-1.02) {Tool check / scope};
  \end{scope}
\end{tikzpicture}}
  \caption{\sys's inference system and two-layer verification. Brackets show
  each verifier's scope; Verus assumes the shared kernel contracts, while our
  kernel verifier proves them.}
  \Description{The center stacks request determinism, the inference engine,
  shared kernel contracts, and GPU kernel source. Two sequences of output events
  interleave the selected requests' blue boxes with other requests' gray boxes.
  Request r1 in trace T1 and request r2 in trace T2 have the same prompt and
  initial sampler state under a fixed model and deployment. Their projected
  views contain three and two outputs, respectively, with identical logit bits
  and tokens on their common prefix. Spacing denotes order, not elapsed time.
  On the left, trusted external APIs and kernel launch code lie outside the
  engine verification scope. The launch code invokes compiled GPU kernels;
  the trusted Triton compiler produces these kernels from the annotated source.
  Kernel contracts require bitwise equality of selected outputs and support
  the engine's invariant-preservation proof and request determinism by induction.
  On the right, an upper orange bracket connects Verus plus Z3 to the inference
  engine, its proof, the induction, and the request-determinism statement,
  with the kernel contracts as assumptions. The contracts enter the engine
  proof through a blue arrow labeled Imported. A lower orange bracket connects our kernel
  verifier plus Z3 to the entire GPU kernels box and the kernel contracts.
  The two brackets end just above and below the vertical midpoint of the
  kernel contracts box, with a small gap separating them.
  Subtitles inside the tool boxes distinguish their roles: Verus uses contracts
  as assumptions, while our kernel verifier checks that kernels satisfy them.
  The verifier takes the annotated Triton source and fixed configuration as
  input and proves the contracts.
  Directly below the contracts, an example compares bitwise-equal selected
  outputs of batched and singleton runs.
  A unified legend shows orange boxes for trusted components,
  neutral boxes for code and proof artifacts, blue boxes for selected request
  data, and gray boxes for other request data. Its arrow samples distinguish
  data flow and execution, projection and correspondence, proof composition,
  and tool checks with their scope brackets. Background colors denote layers,
  not trust.}
  \label{fig:overview}
\end{figure*}

\sys is a deterministic LLM inference system that produces bitwise-identical logits for a
request across the system's execution variations. \autoref{fig:overview}
shows \sys's implementation and proof composition; brackets mark each
verifier's scope. The inference engine schedules and batches tokens, manages
cached prefixes, and invokes GPU kernels.

\sys follows one design rule: \textit{for a fixed model and deployment
configuration, the runtime may vary token batching, state storage, or state
access, while preserving each position's floating-point operation sequence}.
This rule does not intrinsically require a unique kernel or tiling. However, rather than proving bitwise equivalence across alternative implementations, \sys pins kernels and launch configurations selected solely based on the fixed model and environment. Machine-checked engine proofs and the kernel verifier then establish that runtime variation (e.g., batch size, KV-cache hit length, and request scheduling) changes neither a position’s logical inputs nor its floating-point operation sequence.
For example, matmul block sizes may vary with the layer's weight dimensions while remaining fixed across runtime batch sizes.

We verify \sys in two layers. At the engine layer, each GPU operation $K$ is
represented by an uninterpreted function $f_K$ that gives the operation's
single-request result. Verus checks that engine steps preserve request and
cache invariants and that each emitted logit vector depends only on the
request's token history, assuming that the result selected for request $r$
is $f_K(x_r)$ for its logical input $x_r$. The proof reasons about request
states, page tables, and cache sharing and updates without
modeling Triton or floating-point arithmetic.
Assuming the kernel contracts, Verus also checks the induction over engine
steps: requests with the same prompt and initial sampler state produce
per-request views that agree on their common prefix, as illustrated at the
top of \autoref{fig:overview}.

Our kernel verifier proves the contracts in \autoref{fig:overview} under
the stated assumptions. For
each kernel, it proves that the selected output of invocation $K(X)$ equals $f_K(x_r)$ when $x_r$ denotes the same logical input.
The linear-projection example at the bottom of \autoref{fig:overview}
illustrates this relation for batched and singleton invocations.
The two invocations may differ in batch layout, query length, whether the
position is computed during prefill or decode, and physical page tables.  The
proof establishes that both invocations perform the same operations, in the
same order, on the selected input.  Configuration checks
then bind the proved launch parameters to the runtime invocation.
The trusted Triton compiler compiles the source into the GPU kernel that
\sys invokes.

The two proofs meet at $f_K$.  The engine proof establishes the logical input
$x_r$ and uses $f_K(x_r)$ without interpreting its arithmetic; the kernel
proof establishes that the deployed GPU invocation returns that result.
Verus checks each call's preconditions and assumes its postconditions through
a trusted import. It does not recheck the kernel proofs; contract translation
and tensor-representation correspondence remain trusted (\autoref{sec:seam}).
Composing the proofs yields the main invariant in \autoref{sec:engine-invariants}.
\autoref{sec:system} describes the engine implementation and proof, while \autoref{sec:kernel} describes the kernel verifier.

\section{Engine Implementation and Verification}
\label{sec:system}

We design our engine and structure our proofs around the intermediate goal of demonstrating that each engine step produces logits determined solely by the token history for any request.
To build up to this, we refer to an engine step as $E\xrightarrow{O}E'$, with the finite partial output map $O$ defined in \autoref{sec:serving-spec}, and write $o_u=O(u)$ with components $o_u.\mathsf{logits}$ and $o_u.\mathsf{token}$.
For any engine state $E$, we let $E.\mathsf{requests}$ map live request identifiers to records, and write $r_u=E.\mathsf{requests}[u]$.
Given this, we want to show that for any engine step $E\xrightarrow{O}E'$, all request IDs $u\in\operatorname{dom}(O)$ must have their logits ($o_u.\mathsf{logits}$) equal to the result of a deterministic function $\operatorname{Forward}$, which depends only on the configuration ($W$) and the history of the request ($r_u.\mathsf{history}$ using the pre-step state $E$, which is equal to the prompt followed by emitted tokens).
$\operatorname{Forward}$ outputs the last-position logits based on a \emph{cold forward} (i.e., from an empty KV-cache).

To prove this inductively, we need to introduce several invariants whose conjunction forms $\operatorname{Inv}(E,W)$.
Therefore, our goal is to prove the following theorem:

\begin{theorem}[Engine-step preservation]
\label{thm:engine-step}
Under the fixed deployment assumptions and imported kernel contracts,
\[
\begin{gathered}
\operatorname{Inv}(E,W)\land E\xrightarrow{O}E'
\\[2pt]
\Longrightarrow\quad
\left[\begin{aligned}
&\operatorname{Inv}(E',W)\\
&{}\land\forall u\in\operatorname{dom}(O),\\
&\quad o_u.\mathsf{logits}
 =\operatorname{Forward}(W,r_u.\mathsf{history})
\end{aligned}\right].
\end{gathered}
\]
\end{theorem}

To prove our engine meets Definition~\ref{def:system-determinism}, we subsequently need to account for the token sampled from the logits (which $\operatorname{Forward}$ says nothing about).
Sampling is a deterministic function of logits and explicit per-request
sampler state.
Two requests with the same prompt and initial sampler state begin with equal histories and sampler states; whenever both reach their next output, \autoref{thm:engine-step} gives equal logits; in turn, sampling gives equal tokens and successor sampler states, so appending those tokens preserves history equality.
Steps without an output do not advance this per-request sequence.
Verus checks the induction over finite executions: corresponding
requests agree at every output index reached by both, regardless of their
interleaving with other requests.
Their views are therefore prefix-comparable.

In the rest of the section, we will first discuss three design choices that support the verification of determinism, followed by the invariants that form $\operatorname{Inv}(E,W)$ and how they are preserved.

\subsection{Engine Implementation}
\label{sec:engine-implementation}

\sys follows vLLM's architecture~\cite{kwon2023pagedattention}, with
continuous batching, paged KV storage, and prefix caching. We implement the engine in Rust with Verus
annotations and invoke Triton kernels through Python. Three implementation
choices support determinism and its
verification.

\parab{Fixed kernel choices.}
At initialization, the engine binds the model to an immutable runtime
configuration containing the selected kernel specializations and tile parameters. These
choices depend on the fixed model and deployment configuration, not the
current batch size, prefill chunk size, or whether a request is prefilling or
decoding.  The inference path performs no runtime autotuning and does not fall
back to configurations outside the verified set.  \autoref{sec:kernel}
describes the kernels and their numerical guarantees.

\parab{Exact prefix matching.}
A cached page is reusable only for the token prefix that produced it.  The
cache manager uses rolling hashes to find candidate pages, then checks their
exact tokens, depth in the prefix, and physical parent page.  A hash collision
therefore cannot justify reuse.  Under memory pressure, the manager evicts
only unreferenced leaves of this parent chain, so a parent identifier cannot
be recycled while a retained child still refers to it.

\parab{Separate cache reuse from generation.}
The scheduler resolves all prefix-cache hits before registering any pages
that the current step will produce.  Those pages can be reused only in later
steps, after their KV values have been computed.  Previously computed shared
or registered prefix pages remain immutable; new KV values are written to
private suffix slots.  This discipline separates the cache contents a step
may reuse from the values it must compute.

These choices do not require identical scheduling decisions across
executions.  Batch composition, scheduling order, and physical page allocation
may vary. The proof must establish that these choices do not change a
request's logits.

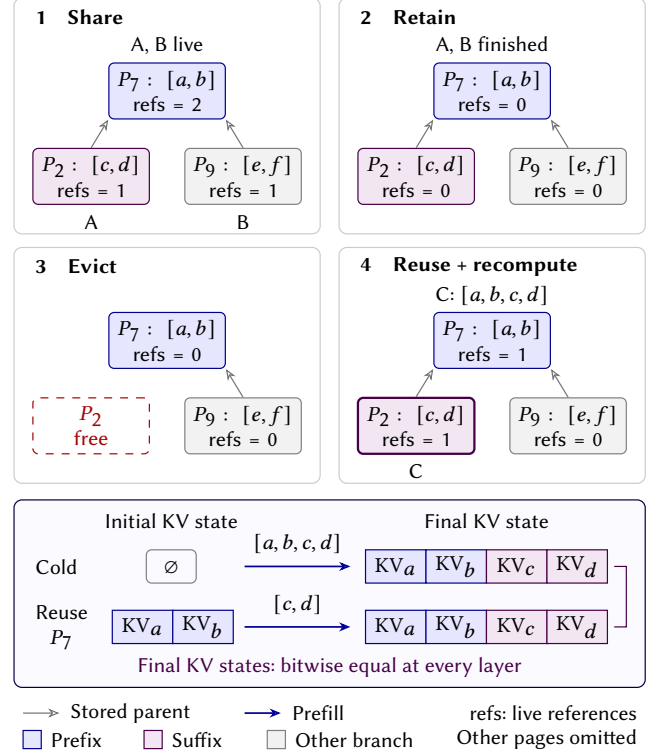
\begin{figure}[t]
  \centering
\begin{tikzpicture}[
  x=1cm, y=1cm,
  font=\sffamily\fontsize{8.05}{9.5}\selectfont,
  line width=0.45pt,
  heading/.style={font=\sffamily\bfseries\fontsize{8.05}{9.5}\selectfont},
  note/.style={font=\sffamily\fontsize{8.05}{9.5}\selectfont, align=center},
  page/.style={draw=black!55, rounded corners=2pt, align=center,
    text width=1.40cm, minimum height=0.74cm, inner sep=2pt},
  prefixpage/.style={page, draw=blue!55!black, fill=blue!10},
  branch/.style={page, draw=violet!55!black, fill=violet!10},
  other/.style={page, draw=black!50, fill=black!5},
  parent/.style={-{Stealth[open,length=1.5mm,width=1.2mm]}, draw=black!55},
  prefill/.style={-{Stealth[open,length=1.5mm,width=1.2mm]},
    draw=blue!60!black, line width=0.65pt}
]
  \useasboundingbox (0,-3.59) rectangle (8.35,6.42);

  \foreach \x in {0,4.30} {
    \foreach \y in {0,3.30} {
      \draw[draw=black!20, rounded corners=3pt]
        (\x,\y) rectangle (\x+4.05,\y+3.12);
    }
  }

  \begin{scope}[yshift=3.30cm]
    \node[heading, anchor=west] at (0.16,2.87) {1\quad Share};
    \node[note] at (2.025,2.49) {A, B live};
    \node[prefixpage] (sharep) at (2.025,1.89)
      {$P_{\mathsf{7}}:\ [a,b]$\\refs = 2};
    \node[branch] (sharea) at (1.02,0.76)
      {$P_{\mathsf{2}}:\ [c,d]$\\refs = 1};
    \node[other] (shareb) at (3.03,0.76)
      {$P_{\mathsf{9}}:\ [e,f]$\\refs = 1};
    \draw[parent] (sharea.north) -- (sharep.south west);
    \draw[parent] (shareb.north) -- (sharep.south east);
    \node[note] at (1.02,0.16) {A};
    \node[note] at (3.03,0.16) {B};
  \end{scope}

  \begin{scope}[xshift=4.30cm,yshift=3.30cm]
    \node[heading, anchor=west] at (0.16,2.87) {2\quad Retain};
    \node[note] at (2.025,2.49) {A, B finished};
    \node[prefixpage] (retainp) at (2.025,1.89)
      {$P_{\mathsf{7}}:\ [a,b]$\\refs = 0};
    \node[branch] (retaina) at (1.02,0.76)
      {$P_{\mathsf{2}}:\ [c,d]$\\refs = 0};
    \node[other] (retainb) at (3.03,0.76)
      {$P_{\mathsf{9}}:\ [e,f]$\\refs = 0};
    \draw[parent] (retaina.north) -- (retainp.south west);
    \draw[parent] (retainb.north) -- (retainp.south east);
  \end{scope}

  \begin{scope}
    \node[heading, anchor=west] at (0.16,2.87) {3\quad Evict};
    \node[prefixpage] (evictp) at (2.025,1.89)
      {$P_{\mathsf{7}}:\ [a,b]$\\refs = 0};
    \node[page, dashed, draw=red!65!black, text=red!65!black]
      (evicta) at (1.02,0.76) {$P_{\mathsf{2}}$\\free};
    \node[other] (evictb) at (3.03,0.76)
      {$P_{\mathsf{9}}:\ [e,f]$\\refs = 0};
    \draw[parent] (evictb.north) -- (evictp.south east);
  \end{scope}

  \begin{scope}[xshift=4.30cm]
    \node[heading, anchor=west] at (0.16,2.87) {4\quad Reuse + recompute};
    \node[note] at (2.025,2.49) {C: $[a,b,c,d]$};
    \node[prefixpage] (reusep) at (2.025,1.89)
      {$P_{\mathsf{7}}:\ [a,b]$\\refs = 1};
    \node[branch, line width=0.85pt] (reusea) at (1.02,0.76)
      {$P_{\mathsf{2}}:\ [c,d]$\\refs = 1};
    \node[other] (reuseb) at (3.03,0.76)
      {$P_{\mathsf{9}}:\ [e,f]$\\refs = 0};
    \draw[parent] (reusea.north) -- (reusep.south west);
    \draw[parent] (reuseb.north) -- (reusep.south east);
    \node[note] at (1.02,0.16) {C};
  \end{scope}

  \draw[draw=blue!25!black, fill=blue!2, rounded corners=3pt]
    (0,-2.70) rectangle (8.35,-0.22);
  \node[note] at (2.08,-0.49) {Initial KV state};
  \node[note] at (6.25,-0.49) {Final KV state};
  \node[note, anchor=west] at (0.16,-1.10) {Cold};
  \node[note, anchor=west] at (0.16,-1.90) {Reuse\\$P_{\mathsf{7}}$};
  \node[draw=black!50, rounded corners=2pt,
    minimum width=0.66cm, minimum height=0.44cm, inner sep=0pt]
    at (2.08,-1.10) {$\varnothing$};

  \foreach \x/\t in {1.30/a,2.10/b} {
    \draw[draw=blue!55!black, fill=blue!10]
      (\x,-2.12) rectangle (\x+0.80,-1.68);
    \node[note] at (\x+0.40,-1.90) {$\mathrm{KV}_{\t}$};
  }
  \foreach \y in {-1.10,-1.90} {
    \foreach \x/\t in {4.65/a,5.45/b} {
      \draw[draw=blue!55!black, fill=blue!10]
        (\x,\y-0.22) rectangle (\x+0.80,\y+0.22);
      \node[note] at (\x+0.40,\y) {$\mathrm{KV}_{\t}$};
    }
    \foreach \x/\t in {6.25/c,7.05/d} {
      \draw[draw=violet!55!black, fill=violet!10]
        (\x,\y-0.22) rectangle (\x+0.80,\y+0.22);
      \node[note] at (\x+0.40,\y) {$\mathrm{KV}_{\t}$};
    }
    \draw[prefill] (3.05,\y) -- (4.45,\y);
  }
  \node[note] at (3.75,-0.80) {$[a,b,c,d]$};
  \node[note] at (3.75,-1.60) {$[c,d]$};
  \draw[draw=violet!55!black]
    (7.94,-1.10) -- (8.10,-1.10) -- (8.10,-1.90) -- (7.94,-1.90);
  \node[note, text=violet!55!black] at (4.175,-2.43)
    {Final KV states: bitwise equal at every layer};

  \draw[parent] (0.12,-3.04) -- (0.60,-3.04);
  \node[note, anchor=west, inner sep=0pt] at (0.74,-3.04) {Stored parent};
  \draw[prefill] (3.05,-3.04) -- (3.53,-3.04);
  \node[note, anchor=west, inner sep=0pt] at (3.67,-3.04) {Prefill};
  \node[note, anchor=east, inner sep=0pt] at (8.25,-3.04) {refs: live references};

  \draw[draw=blue!55!black, fill=blue!10] (0.12,-3.51) rectangle (0.36,-3.27);
  \node[note, anchor=west, inner sep=0pt] at (0.48,-3.39) {Prefix};
  \draw[draw=violet!55!black, fill=violet!10] (1.72,-3.51) rectangle (1.96,-3.27);
  \node[note, anchor=west, inner sep=0pt] at (2.08,-3.39) {Suffix};
  \draw[draw=black!50, fill=black!5] (3.35,-3.51) rectangle (3.59,-3.27);
  \node[note, anchor=west, inner sep=0pt] at (3.71,-3.39) {Other branch};
  \node[note, anchor=east, inner sep=0pt] at (8.25,-3.39) {Other pages omitted};
\end{tikzpicture}
  \caption{The KV-cache invariant across sharing, retention, eviction, and
  recomputation: every initialized live or retained entry equals the KV value from a cold
  forward on its corresponding causal prefix. $P_i$ denotes physical KV-cache
  page $i$; each page holds KV values for two tokens in this example.}
  \Description{Four snapshots in a two-by-two grid, read left to right and
  then top to bottom, show registered two-token prompt pages. Requests
  A and B share page P7 for tokens a and b, with separate suffix pages P2 for c
  and d and P9 for e and f. After both requests finish, all three pages remain
  cached with zero live references. Memory pressure evicts leaf P2, but P7
  cannot be evicted while retained child P9 names it as a parent. Request C
  then reuses P7 and recomputes c and d in a new allocation of P2. Sharing and
  retention preserve KV equality; eviction removes the obligation for P2 while
  preserving it for P7 and P9; recomputation re-establishes it for P2. Two
  aligned rows compare initial and final logical KV states, with both prefill
  arrows pointing left to right. The cold row starts with an empty cache and
  prefills a, b, c, d. The reuse row starts with cached KV for a and b in P7
  and prefills c and d. Each row ends with four cells for a, b, c, d;
  a bracket and a statement beneath both rows mark the final states as bitwise
  equal at every layer. These cells
  show KV values, not the logits returned by prefill. Under the same model and deployment,
  these values are bitwise equal at every layer, regardless of physical page
  allocation. Blue cells identify the prefix and purple cells the suffix.
  P7 and P2 hold these values in snapshots 1, 2, and 4; there is no content
  obligation for free P2 in snapshot 3. Gray arrows point to stored parents;
  blue arrows show prefill transitions.
  Private tails and unrelated pages are omitted.}
  \label{fig:cache-lifecycle}
\end{figure}

\subsection{Engine Invariants}
\label{sec:engine-invariants}

To establish request determinism across batching, scheduling, and cache reuse,
we maintain a central KV-cache invariant alongside request, layout, and
write-isolation invariants. The KV-cache invariant ensures that every live or
retained reusable KV entry equals the value obtained by recomputing its token
prefix without cache.
Our mechanized proof uses ghost state representing logical request histories
and cache provenance; this state does not affect runtime.

\parab{Logical values and notation.}
$\operatorname{Forward}(W,h)$, for a non-empty request history $h$, composes the fixed single-request uninterpreted kernel functions
$f_K$ introduced in \autoref{sec:overview}.
For $0\leq p<|h|$,
define $f_{\mathsf{KV}}(W,h,p)$ as the collection of KV pairs across all decoder layers
produced at position $p$ by a cold forward on the first $p+1$ tokens of $h$.
We call this the \emph{canonical} KV value. It depends only on the causal
prefix ending at $p$; the kernel contracts establish that processing later
tokens does not change this value.
The fields $r_u.\mathsf{sampler}$ and $r_u.\mathsf{cached}$ hold the request's
current sampler state and initialized-cache length, respectively.
The lookup $r_u.\mathsf{slot}(p)$ maps logical
position $p$ to a physical slot through its page table, and
$E.\mathsf{kv}[s]$ denotes the collection of KV pairs across all layers at
slot $s$. This notation groups values logically, without requiring contiguous
storage across layers. Equality of KV values means bitwise equality at every
layer; logit equality is also bitwise.

\parab{Request and layout consistency.}
Queues and request maps agree on live identifiers, and request records track
their histories, sampler states, and stopping states consistently. Initialized
cache lengths satisfy $0\leq r_u.\mathsf{cached}\leq |r_u.\mathsf{history}|$,
and page tables map these
prefixes to allocated, in-range slots. Tensor shapes, permissions, capacity,
and the fixed runtime configuration satisfy the engine's readiness conditions.
Planning establishes the corresponding kernel-call premises, including the
mapping from each packed row to its request and logical token position.

\parab{KV-value consistency.}
Let $E.\mathsf{prefixes}$ collect initialized live-request prefixes (including private tails) and retained reusable prefixes. Each prefix record
$q$ has a token sequence $q.\mathsf{tokens}$ and an equally long sequence
$q.\mathsf{slots}$ of physical slots in token order.
These records are ghost views of request and page metadata.

Write $h[0:c]$ for the first $c$ tokens of a history $h$. For every live
request identifier $u$, the collection contains a record $q_u$ satisfying
\[
\begin{aligned}
q_u.\mathsf{tokens}
&=r_u.\mathsf{history}[0:r_u.\mathsf{cached}],\\
q_u.\mathsf{slots}[p]&=r_u.\mathsf{slot}(p)
\quad(0\leq p<r_u.\mathsf{cached}).
\end{aligned}
\]
The collection also contains a record for every resident chain with published
prefix provenance: its tokens and slots are reconstructed from exact token
blocks, logical depths, and parent links. This coverage includes chains with
no live references and those not currently selected by the hash index.

For every $q\in E.\mathsf{prefixes}$ and $0\leq p<|q.\mathsf{tokens}|$,
the KV-value invariant requires
\[
E.\mathsf{kv}[q.\mathsf{slots}[p]]
=f_{\mathsf{KV}}(W,q.\mathsf{tokens},p).
\]
The coverage conditions and this equality are part of $\operatorname{Inv}$.
Together they give canonical values for every initialized request position
and every retained prefix, independently of whether its creating request is
still live. Uninitialized slots are not readable and carry no KV-value
obligation.

\autoref{fig:cache-lifecycle} illustrates this persistence. Requests A and B
share $P_7$ and retain separate suffix pages $P_2$ and $P_9$. Their KV-value
obligations survive request completion. Evicting leaf $P_2$ removes its
obligation without affecting $P_7$ or $P_9$; $P_7$ cannot be evicted while
$P_9$ refers to it. A later request reuses $P_7$ and recomputes the suffix,
re-establishing the same canonical values in newly allocated storage.

\parab{Sharing and write isolation.}
Shared and registered prefix pages are immutable, and writable suffix slots
are private. Each execution plan assigns in-range, mutually disjoint write
slots that do not overlap other requests' initialized prefixes or retained
values. Consequently, extending one request's cache preserves the values
available to other requests.

\subsection{Invariant Preservation}
\label{sec:engine-preservation}

\parab{Initialization and admission.}
Initialization establishes $\operatorname{Inv}$: request records contain their
initial histories and sampler states. Initialized-cache lengths are zero and
no prefixes are retained, so the KV-value invariant holds vacuously. Admission
adds a request in this state at a stable engine boundary, preserving existing
cache values and the invariant.

\parab{Planning and forward execution.}
Planning may combine cold prefill, cached partial prefill, and decode work.
It resolves prefix hits, allocates private suffix storage, and evicts only
unreferenced leaves. Row mappings and write-slot checks establish the
kernel-contract premises (\autoref{sec:kernel}).

For a cache hit, let $q\in E.\mathsf{prefixes}$ represent the retained chain
being reused, and let $c=|q.\mathsf{tokens}|$. Exact prefix matching establishes
\[
r_u.\mathsf{history}[0:c]=q.\mathsf{tokens}.
\]
Let $E^+$ denote the intermediate state after attaching the prefix, and let
$r^+=E^+.\mathsf{requests}[u]$. Attachment preserves the request's history
and the reused values, and maps each position $p<c$ to $q.\mathsf{slots}[p]$.
Consequently, for every $0\leq p<c$,
\[
\begin{aligned}
E^+.\mathsf{kv}[r^+.\mathsf{slot}(p)]
&=E.\mathsf{kv}[q.\mathsf{slots}[p]]\\
&=f_{\mathsf{KV}}(W,q.\mathsf{tokens},p)\\
&=f_{\mathsf{KV}}(W,r^+.\mathsf{history},p).
\end{aligned}
\]
The KV-value invariant gives the second equality; dependence only on the
causal prefix gives the third. Thus, adding the reused prefix to the request's
initialized cache preserves consistency, even if no earlier request remains
live.

The forward proof proceeds by decoder layer and scheduled token position.
The kernel contracts establish cold-forward equality for cold prefill.
For cached work, the KV invariant supplies a canonical prefix, and a
continuation lemma establishes that every processed position yields the same logits and new KV values that a cold
forward would. Layout mappings connect these logical
values to packed tensors and physical cache slots. When a request emits an
output, its selected row corresponds to the last position of its pre-step
history, yielding the output equality in \autoref{thm:engine-step}.

\parab{State updates.}
Newly written KV values are canonical; write isolation preserves the values
of other requests and retained pages. Extending a request's initialized prefix
extends its record in the prefix collection. Publishing a newly computed full
page also establishes a retained-prefix record for these canonical values,
so their KV-value obligation survives the creating request's completion.
Leaf-only eviction preserves the obligations of surviving pages.
After sampling for request $u$, the engine appends the emitted token, updates
its sampler state and cache length, then adjusts queues and reference
counts. Finished requests may be removed without discarding retained pages'
KV-value obligations. Steps without an output for a request preserve its
history and sampler state but may extend its initialized cache. Together,
these updates re-establish $\operatorname{Inv}$ at the next engine-step boundary.

\section{Kernel Implementation and Verification}
\label{sec:kernel}

Kernel contracts require equal logical inputs and guarantee bitwise-equal
selected outputs despite changes in batch size, query length, or KV
cache hits. Equal inputs alone do not ensure this property: changes in the
operation sequence can affect floating-point rounding. Our verifier establishes
two facts: corresponding computations apply matching value-relevant operations
in the same order, and those operations receive equal inputs.
We develop the structural analysis using matmul, then extend it with value
analysis for attention causality, under the verification assumptions stated below.

\subsection{Relational Kernel Contracts}
\label{sec:kernel-annotations}
\label{sec:seam}

Each contract relates input and output regions of two kernel invocations. Regions select tensor cells by
logical indices; specified equality is bitwise and other inputs need not agree.

For example, consider matmul with $M\times K$ activations $a$, $K\times N$
weights $b$, and $M\times N$ output $c$. Each row in $a$ represents one
token; a request may contribute multiple rows. The contract compares any
selected token row with its singleton run, requiring equal inputs
for that row and the weights. All
other token rows remain unconstrained.

Kernel annotations are comments and do not affect execution.
\texttt{@params} declares tensor shapes, strides, element kinds, and typed
scalars; \texttt{@grid} declares the launch grid; and \texttt{@verif} states a
relation between invocations.
The verifier treats \texttt{@params} and \texttt{@grid} as the invocation model
and \texttt{same} and \texttt{pre} as premises, and proves the output equality
specified by \texttt{post}.
\autoref{fig:matmul-analysis} combines the annotation, tile IR, and dependency
analysis, with callouts identifying the obligations checked by Z3.
We abbreviate strides and use mathematical notation. With fixed tile
dimensions $B_M,B_N,B_K$, ceiling division covers the output with
$B_M\times B_N$ tiles.

Here, subscripts $L$ and $R$ denote the packed-token and singleton runs.
The free integer variable $x$ is automatically
universally quantified: the contract holds for every selected row satisfying
$0\le x<M_L$ whenever the remaining premises hold.
Both invocations share weight dimensions and tile constants.
This row-level guarantee supports request-level batch invariance
and independence from how tokens are grouped into matmul invocations.

The engine proof establishes kernel-call preconditions and imports kernel
contracts as trusted assumptions in Verus. Contract translation and the
correspondence between tensor representations remain trusted.

\subsection{Relational Structural Analysis}
\label{sec:z3-proof}

\begin{figure}[t]
\centering
\begin{tikzpicture}[
  x=1cm, y=1cm,
  font=\sffamily\fontsize{8.05}{9.5}\selectfont,
  line width=0.45pt,
  heading/.style={font=\sffamily\bfseries\fontsize{8.05}{9.5}\selectfont},
  code/.style={anchor=west, inner sep=0pt,
    font=\sffamily\fontsize{8.05}{9.5}\selectfont},
  note/.style={inner sep=1pt},
  other/.style={draw=black!35, fill=black!5},
  needed/.style={draw=blue!65!black, fill=blue!14},
  badge/.style={circle, draw=blue!65!black, text=blue!65!black,
    fill=white, minimum size=2.8mm, inner sep=0pt, outer sep=0pt,
    line width=0.35pt,
    font=\sffamily\fontsize{8.05}{9.5}\selectfont},
  demand/.style={-{Stealth[open,length=1.5mm,width=1.2mm]},
    draw=blue!65!black, line width=0.6pt},
  solver/.style={draw=black!50, rounded corners=2pt, inner sep=3pt},
  check/.style={-{Stealth[open,length=1.5mm,width=1.2mm]},
    draw=black!65, dashed, line width=0.55pt}
]
  \definecolor{mmDecorator}{HTML}{A02040}
  \definecolor{mmKeyword}{HTML}{A02040}
  \definecolor{mmFunction}{HTML}{A04000}
  \definecolor{mmContract}{HTML}{008000}
  \newcommand{\mmdecorator}[1]{\textcolor{mmDecorator}{\texttt{#1}}}
  \newcommand{\mmkeyword}[1]{\textcolor{mmKeyword}{\texttt{#1}}}
  \newcommand{\mmfunction}[1]{\textcolor{mmFunction}{\texttt{#1}}}
  \newcommand{\mmcontract}[1]{\textcolor{mmContract}{\texttt{#1}}}
  \useasboundingbox (0,-0.93) rectangle (8.25,13.06);

  \begin{scope}[yshift=-0.95cm]
  \node[note, anchor=west] at (0,13.80)
    {Kernel runs: $L$: arbitrary $M>0$; $R$: $M=1$.};
  \node[note, anchor=west] at (0,13.37)
    {Compare row $x$ ($0\le x<M$) in $L$ with row $0$ in $R$.};

  \draw[black!20, rounded corners=2pt] (0,8.70) rectangle (8.25,13.14);
  \node[code] at (0.18,12.88) {\mmdecorator{@params}\texttt{(}};
  \node[code] at (0.48,12.48)
    {\mmfunction{tensor}\texttt{(}$a$\texttt{, }\mmfunction{float}\texttt{, }%
     \mmfunction{shape}\texttt{(}$M,K$\texttt{), }\mmfunction{strides}\texttt{(...)),}};
  \node[code] at (0.48,12.08)
    {\mmfunction{tensor}\texttt{(}$b$\texttt{, }\mmfunction{float}\texttt{, }%
     \mmfunction{shape}\texttt{(}$K,N$\texttt{), }\mmfunction{strides}\texttt{(...)),}};
  \node[code] at (0.48,11.68)
    {\mmfunction{tensor}\texttt{(}$c$\texttt{, }\mmfunction{float}\texttt{, }%
     \mmfunction{shape}\texttt{(}$M,N$\texttt{), }\mmfunction{strides}\texttt{(...)))}};
  \node[code] at (0.18,11.28)
    {\mmdecorator{@grid}\texttt{(}$\lceil M/B_M\rceil,\ \lceil N/B_N\rceil$\texttt{)}};
  \node[code] at (0.18,10.88) {\mmdecorator{@verif}\texttt{(batch\_invariance,}};
  \node[code] at (0.48,10.48) {\mmcontract{same}\texttt{(}$N,K$\texttt{),}};
  \node[code] at (0.48,10.08)
    {\mmcontract{pre}\texttt{(}$0\le x<M_L,\ M_R=1$,};
  \node[code] at (0.88,9.68)
    {$a_L[x:x+1,\ 0:K]=a_R[0:1,\ 0:K]$,};
  \node[code] at (0.88,9.28)
    {$b_L[0:K,\ 0:N]=b_R[0:K,\ 0:N]$\texttt{),}};
  \node[code] at (0.48,8.88)
    {\mmcontract{post}\texttt{(}$c_L[x:x+1,\ 0:N]=c_R[0:1,\ 0:N]$\texttt{))}};
  \end{scope}

  \begin{scope}[yshift=-0.38cm]
  \draw[black!20, rounded corners=2pt] (0,3.64) rectangle (8.25,8.00);
  \draw[solver, dashed] (4.85,5.40) rectangle (8.12,7.92);
  \node[heading, anchor=west, text=black!70] at (4.95,5.63)
    {Z3: pairing and bounds};
  \node[code] at (0.18,7.74)
    {\mmkeyword{parfor}\texttt{\space}$i$\texttt{\space}\mmkeyword{in}\texttt{\space}\mmfunction{range}\texttt{(}$\lceil M/B_M\rceil$\texttt{):}};
  \node[code] at (0.48,7.31)
    {\mmkeyword{parfor}\texttt{\space}$j$\texttt{\space}\mmkeyword{in}\texttt{\space}\mmfunction{range}\texttt{(}$\lceil N/B_N\rceil$\texttt{):}};
  \node[code] at (0.78,6.88) {$m\gets iB_M,\quad n\gets jB_N$};
  \node[code] at (0.78,6.45)
    {$S\gets\mathop{\text{\mmfunction{zeros}}}(B_M,B_N)$};
  \node[code] at (0.78,6.02)
    {\mmkeyword{for}\texttt{\space}$k$\texttt{\space}\mmkeyword{in}\texttt{\space}\mmfunction{range}\texttt{(}$\lceil K/B_K\rceil$\texttt{):}};
  \node[code] at (1.08,5.59) {$t\gets kB_K$};
  \node[code] at (1.08,5.16)
    {$A\gets a[m:m+B_M,\ t:t+B_K]$};
  \node[code] at (1.08,4.73)
    {$B\gets b[t:t+B_K,\ n:n+B_N]$};
  \node[code] at (1.08,4.30)
    {$S\gets\mathop{\text{\mmfunction{mma}}}(A,B,S)$};
  \node[code] at (0.78,3.88)
    {$c[m:m+B_M,\ n:n+B_N]\gets S$};
  \foreach \badgeid/\badgeheight in {1/5.16,2/4.73,3/4.30,4/3.88} {
    \node[badge] at (0.30,\badgeheight) {\badgeid};
  }
  \node[code, text=black!65] at (4.95,7.74)
    {$i_L=\lfloor x/B_M\rfloor\leftrightarrow i_R=0$};
  \node[code, text=black!65] at (4.95,7.31)
    {$j_L=j_R$};
  \node[code, text=black!65] at (4.95,6.02)
    {$k_L=k_R$ in order};
  \end{scope}

  \node[note, anchor=west] at (0,2.94)
    {Rows within tiles: $r_L=x\bmod B_M$; $r_R=0$.};

  \begin{scope}[yshift=0.37cm]
  \foreach \base/\bottom/\label in {2.30/1.08/A,2.30/-0.22/B,4.65/0.245/S} {
    \node[note] at (\base+0.60,\bottom+1.06) {$\label$};
    \foreach \col in {0,1,2} {
      \foreach \row in {0,1,2} {
        \draw[other] (\base+\col*0.40,\bottom+\row*0.27)
          rectangle ({\base+(\col+1)*0.40},{\bottom+(\row+1)*0.27});
      }
    }
  }
  \foreach \col in {0,1,2} {
    \draw[needed] (2.30+\col*0.40,1.35)
      rectangle ({2.30+(\col+1)*0.40},1.62);
    \foreach \row in {0,1,2} {
      \draw[needed] (2.30+\col*0.40,-0.22+\row*0.27)
        rectangle ({2.30+(\col+1)*0.40},{-0.22+(\row+1)*0.27});
    }
    \draw[needed] (4.65+\col*0.40,0.515)
      rectangle ({4.65+(\col+1)*0.40},0.785);
  }
  \node[needed, minimum width=1.80cm, minimum height=0.52cm, inner sep=2pt]
    (aregion) at (0.95,1.485) {$a[m+r,\ldots]$};
  \node[needed, minimum width=1.80cm, minimum height=0.52cm, inner sep=2pt]
    (bregion) at (0.95,0.185) {$b[\ldots,\ldots]$};
  \node[needed, minimum width=1.80cm, minimum height=0.52cm, inner sep=2pt]
    (cregion) at (7.15,0.65) {$c[m+r,\ldots]$};
  \draw[demand] (cregion.west) -- (5.90,0.65);
  \draw[blue!65!black, line width=0.6pt] (4.60,0.65) -- (4.10,0.65);
  \draw[demand] (4.10,0.65) |- (3.55,1.485);
  \draw[demand] (4.10,0.65) |- (3.55,0.185);
  \draw[demand] (2.25,1.485) -- (aregion.east);
  \draw[demand] (2.25,0.185) -- (bregion.east);
  \draw[demand] (5.65,1.10)
    .. controls (6.20,2.15) and (4.30,2.15) .. (4.85,1.10);
  \node[badge] at (2.085,1.83) {1};
  \node[badge] at (2.085,0.53) {2};
  \node[badge] at (4.10,0.99) {3};
  \node[badge] at (5.25,1.89) {3};
  \node[badge] at (6.075,1.02) {4};

  \node[note, align=center, text=black!70] at (7.15,1.83)
    {\textbf{Z3:} stores cover\\\mmcontract{post}};
  \draw[check] (7.15,1.44) -- (cregion.north);
  \draw[solver, dashed] (0,-0.19) rectangle (1.90,1.87);
  \node[solver, minimum width=3.80cm, minimum height=0.48cm]
    (premises) at (2.00,-0.75) {\mmcontract{pre}: input equalities};
  \draw[check] (premises.north -| aregion.center) -- (0.95,-0.19);
  \node[note, anchor=west, align=left, text=black!70] at (4.08,-0.75)
    {\textbf{Z3:} input-region coverage\\and matching load masks};
  \end{scope}
\end{tikzpicture}
\caption{Matmul contract, tile IR, and Z3 checks. Numbers link
statements to backward dependencies. Slices are half-open;
loads are zero-padded and stores bounds-masked. $\operatorname{mma}$
denotes \texttt{tl.dot}; \texttt{parfor} denotes independent grid iterations.}
\label{fig:matmul-analysis}
\Description{The annotation declares tensor shapes, an output-tile grid,
equal input regions, and the required output equality.
The left kernel run has arbitrary positive M; the right kernel
run has M equal to one. Two outer parallel loops enumerate the row and column axes of
the launch grid. Each program initializes an accumulator, loads activation
and weight tiles in a reduction loop, applies a multiply--accumulate update, and stores the
output. The two runs pair the same ordered reduction iterations. The left
program contains token row x and the right program contains singleton row
zero. A backward dependency diagram starts from a region of output tensor c,
reaches the compared accumulator row S, and branches through activation tile
A and weight tile B to the required regions of input tensors a and b.
The compared output, accumulator row, and required inputs are blue.
A self-loop on the accumulator denotes dependence on the same row in its
preceding state. Other rows remain unconstrained. Z3 callouts identify program
pairing and loop-bound checks, coverage of declared outputs by stores, and
coverage of the required activation and weight regions by annotated input
equalities with matching load masks. The prior accumulator row is related by
loop induction rather than by the input annotation.}
\end{figure}
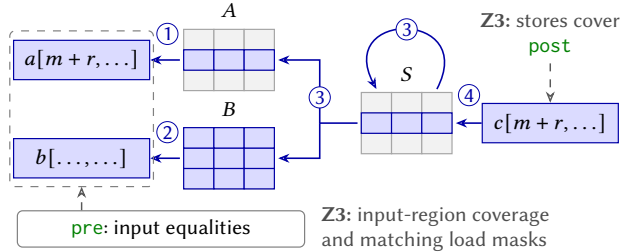

\parab{Tile representation.}
\label{sec:tile-ir}
Our kernels use a restricted subset of Triton, compiled by the
unmodified Triton compiler. The verifier specializes tile constants and
translates block pointers, masked accesses, tile operations, structured loops,
and scalar branches into a typed tile IR. Block-pointer geometry and scalar
metadata identify tensor regions. Unsupported constructs and unresolved
shapes are rejected.

For matmul (\autoref{fig:matmul-analysis}), each program starts with a zero FP32 accumulator of shape
$B_M\times B_N$. It loads activation and weight
tiles, accumulates their products over
$\lceil K/B_K\rceil$ iterations, then stores the
output tile. Tile-size selection depends on fixed $N$ and $K$, not runtime $M$.

\parab{Relational control-flow alignment.}
The verifier seeks a \emph{control-flow isomorphism} pairing program instances
and operations relevant to the selected outputs. Straight-line operations pair in source
order with matching operators, types, static parameters, and operand order.
Branches require equal value-relevant guards and recursively aligned paths.
Reduction loops pair iterations in the same order, requiring equal ranges,
related initial state, and corresponding state updates. Parallel grid loops
(\texttt{parfor} in \autoref{fig:matmul-analysis}) may instead pair only the
independent iterations producing the selected regions.
Z3 checks the scalar, guard, and range obligations under the contract premises.

In matmul, the verifier pairs the program containing left token row $x$ with
the singleton program containing right row zero, for each output-column tile.
Shared $K$ and $B_K$ give equal loop ranges. The zero
initializations, accumulation operations, and stores align,
while the selected rows may have different offsets within their tiles.

\parab{Backward rules.}
Along this correspondence, backward rules determine which operand regions
must agree given output regions. Elementwise operations require corresponding regions; broadcasts
map them to source coordinates; reductions extend them over the reduced axis.
For a matrix-product tile $D=\texttt{tl.dot}(A,B)$, let $I$ and $J$ select
output rows and columns, and let $H$ be the complete contracted index range.
Writing $A[I,H]$ for the subregion selected by these indices, the rule is
$D[I,J]\leftarrow A[I,H],\ B[H,J]$.
The arrow maps a selected output region to its required input regions;
rows of $A$ outside $I$ remain unconstrained. For $\operatorname{mma}$, the
required inputs additionally include the corresponding region of the incoming
accumulator. Z3 checks that annotated input equalities cover
these dependencies, relevant load masks agree, and stores cover the declared
outputs.

Applying these rules to matmul, the selected output row requires the
corresponding final accumulator row. As \autoref{fig:matmul-analysis}
illustrates, each accumulation step requires
its preceding accumulator row, the selected activation-row segment, and the
corresponding weight tile. The annotated input equalities cover those segments
and tiles. Starting from equal zero accumulators, induction over the aligned
iterations gives equal accumulator rows; the paired stores then
give the contract's output equality under the trusted tile-operation rules.

\subsection{Value-Aware Regional Analysis}
\label{sec:pagedattention}

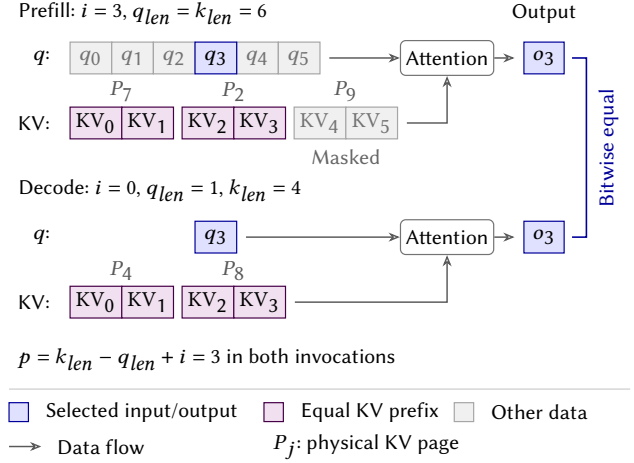
\begin{figure}[t]
\centering
\begin{tikzpicture}[
  x=0.95cm, y=1cm,
  font=\sffamily\fontsize{8.05}{9.5}\selectfont,
  line width=0.45pt,
  heading/.style={font=\sffamily\bfseries\fontsize{8.05}{9.5}\selectfont},
  note/.style={font=\sffamily\fontsize{8.05}{9.5}\selectfont},
  selected/.style={draw=blue!60!black, fill=blue!12},
  shared/.style={draw=violet!55!black, fill=violet!12},
  other/.style={draw=black!35, fill=black!6},
  op/.style={draw=black!55, rounded corners=2pt,
    minimum width=1.05cm, minimum height=0.46cm, inner sep=2pt},
  flow/.style={-{Stealth[open,length=1.5mm,width=1.2mm]},
    draw=black!65, line width=0.55pt},
  equal/.style={draw=blue!60!black, line width=0.6pt}
]
  \useasboundingbox (0,-1.05) rectangle (8.65,5.10);

  \begin{scope}
    \node[note, anchor=west] at (0,4.87)
      {Prefill: $i=3$, $q_{\mathit{len}}=k_{\mathit{len}}=6$};
    \begin{scope}[yshift=0.30cm]
    \node[note, anchor=west] at (0,2.23)
      {Decode: $i=0$, $q_{\mathit{len}}=1$, $k_{\mathit{len}}=4$};

    \foreach \j in {0,1,2,4,5} {
      \draw[other] (0.85+\j*0.58,3.77)
        rectangle (1.43+\j*0.58,4.19);
      \node[note, text=black!55] at (1.14+\j*0.58,3.98) {$q_{\mathsf{\j}}$};
    }
    \draw[selected] (2.59,3.77) rectangle (3.17,4.19);
    \node[note] at (2.88,3.98) {$q_{\mathsf{3}}$};
    \draw[selected] (2.59,1.40) rectangle (3.17,1.82);
    \node[note] at (2.88,1.61) {$q_{\mathsf{3}}$};
    \node[note, anchor=east] at (0.70,3.98) {$q$:};
    \node[note, anchor=east] at (0.70,1.61) {$q$:};

    \node[op] (atleft) at (6.10,3.98) {Attention};
    \node[op] (atright) at (6.10,1.61) {Attention};
    \draw[flow] (4.46,3.98) -- (atleft.west);
    \draw[flow] (3.30,1.61) -- (atright.west);

    \foreach \base/\j in {0.85/0,1.57/1,2.41/2,3.13/3} {
      \foreach \y in {3.12,0.75} {
        \draw[shared] (\base,\y-0.21) rectangle (\base+0.72,\y+0.21);
        \node[note] at (\base+0.36,\y) {$\mathrm{KV}_{\mathsf{\j}}$};
      }
    }
    \foreach \base/\j in {3.97/4,4.69/5} {
      \draw[other] (\base,2.91) rectangle (\base+0.72,3.33);
      \node[note, text=black!55] at (\base+0.36,3.12) {$\mathrm{KV}_{\mathsf{\j}}$};
    }
    \foreach \x/\p in {1.57/7,3.13/2,4.69/9} {
      \node[note, text=black!65] at (\x,3.54) {$P_{\mathsf{\p}}$};
    }
    \foreach \x/\p in {1.57/4,3.13/8} {
      \node[note, text=black!65] at (\x,1.17) {$P_{\mathsf{\p}}$};
    }
    \node[note, anchor=east] at (0.70,3.12) {KV:};
    \node[note, anchor=east] at (0.70,0.75) {KV:};
    \node[note, text=black!60] at (4.69,2.68) {Masked};
    \draw[flow] (5.54,3.12) -| (atleft.south);
    \draw[flow] (3.98,0.75) -| (atright.south);

    \foreach \y/\n in {3.98/atleft,1.61/atright} {
      \draw[flow] (\n.east) -- (7.04,\y);
      \draw[selected] (7.17,\y-0.21) rectangle (7.73,\y+0.21);
      \node[note] at (7.45,\y) {$o_{\mathsf{3}}$};
    }
    \node[note] at (7.45,4.57) {Output};
    \draw[equal] (7.84,3.98) -- (8.04,3.98) -- (8.04,1.61)
      -- (7.84,1.61);
    \node[note, rotate=90, text=blue!60!black] at (8.35,2.79)
      {Bitwise equal};
    \end{scope}
    \node[note, anchor=west] at (0,0.30)
      {$p=k_{\mathit{len}}-q_{\mathit{len}}+i=3$ in both invocations};
  \end{scope}

  \draw[black!20] (0,-0.10) -- (8.65,-0.10);
  \draw[selected] (0,-0.53) rectangle (0.27,-0.28);
  \node[note, anchor=west] at (0.38,-0.405) {Selected input/output};
  \draw[shared] (3.55,-0.53) rectangle (3.82,-0.28);
  \node[note, anchor=west] at (3.93,-0.405) {Equal KV prefix};
  \draw[other] (6.20,-0.53) rectangle (6.47,-0.28);
  \node[note, anchor=west] at (6.58,-0.405) {Other data};
  \draw[flow] (0,-0.86) -- (0.45,-0.86);
  \node[note, anchor=west] at (0.54,-0.86) {Data flow};
  \node[note, anchor=west] at (3.55,-0.86) {$P_j$: physical KV page};
\end{tikzpicture}
\caption{Prefill and decode use equal queries and KV prefixes despite different query indices and physical pages, yielding bitwise-equal outputs. Here, $q_p$, $\mathrm{KV}_p$, and $o_p$ are indexed by logical position. Future KV positions are causally masked.}
\label{fig:implicit-batch}
\Description{Attention compares query index three in a six-token prefill with query index zero in a decode with four key positions. Both select logical position three and have equal query values and KV values through that position. The prefill uses physical KV pages seven and two for this prefix; the decode uses pages four and eight. Future KV positions four and five in the prefill are masked. The contract requires bitwise equality of the selected outputs.}
\end{figure}

Structural analysis alone cannot establish attention causality: it
conservatively retains dependencies on masked KV cells and requires matching
loop ranges. Value analysis can remove these dependencies and establish that
extra key-tile iterations preserve the selected state.

\parab{The causal contract.}
Let $q_{\mathit{len}}$ and $k_{\mathit{len}}$ be an invocation's query and key
lengths, with $0<q_{\mathit{len}}\le k_{\mathit{len}}$. Query row $i$, where
$0\le i<q_{\mathit{len}}$, has logical position
\[
p = k_{\mathit{len}} - q_{\mathit{len}} + i.
\]
The contract requires equal selected queries at equal $p$ and equal logical
KV values through $p$, with shared head geometry, scale, page size, and tile
constants (\autoref{fig:implicit-batch}). Lengths, row indices, and physical
page mappings may differ.

\parab{Forward value analysis.}
The kernel verifier uses conservative abstract interpretation~\cite{cousot1977abstract} to track regions known to
contain zero, one, or negative infinity, Boolean mask values, and unchanged
state. Unknown positions remain unconstrained. This information refines
the backward dependencies: a false mask selects its false branch; masked
scores become negative infinity; subsequent numerical rules identify zeroed tensor regions and restrict the following matrix product's value dependencies.
Rules may carry explicit side conditions, such as finiteness of affected values.

\parab{From value analysis to causal equivalence.}
Shared key-tile iterations align in order, with selected KV dependencies
confined to the causal prefix. Extra iterations have all-false masks for the
selected row, for which the kernel explicitly retains its incoming state.
The value analysis establishes this exact identity transition. We relax
\autoref{sec:z3-proof}'s equal-range requirement only for iterations proved to
preserve the selected state bitwise, extending control-flow alignment
\emph{up to identity steps}. Equal initial states and aligned updates then give equal
outputs without cancellation or reassociation. \autoref{app:attention-verification} develops the attention contract, tile IR,
and relational argument in more detail.

\begin{figure*}[t]
  \centering
  \includegraphics[width=\textwidth]{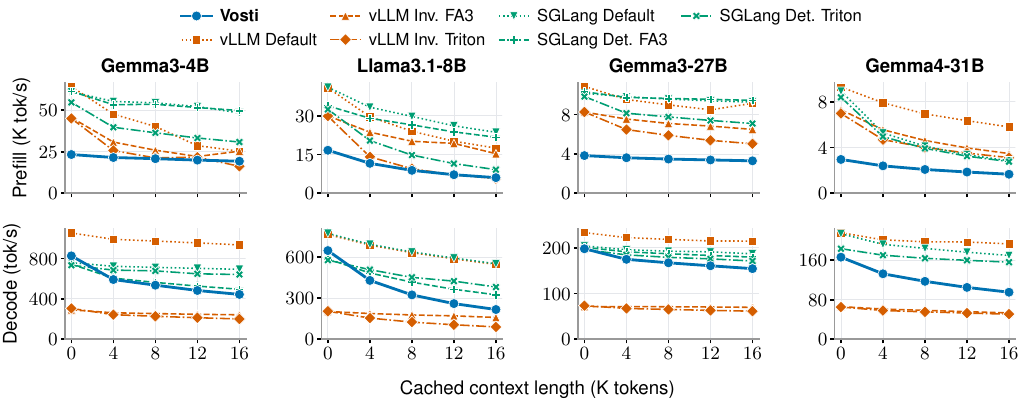}
  \caption{Prefill (top) and decode (bottom) throughput (higher is better)
  versus cached context length.}
  \label{fig:cached-phases}
  \Description{Eight line plots arranged in two rows and four model columns,
  with prefill above decode, comparing seven engine configurations at cached contexts
  of 0, 4096, 8192, 12288, and 16384 tokens. Prefill
  rates are in thousands of tokens per second and decode rates in tokens per
  second. \sys exceeds both vLLM invariant modes in decode throughput at
  every tested context, while prefill remains slower than invariant FA3.}
\end{figure*}

\parab{Verification assumptions.}
The verifier checks IR-level rule applications and relational obligations.
It assumes that compiled Triton operations preserve the stated regional
dependencies, operation correspondences, and numerical transfer rules;
tests support these assumptions on the evaluated configurations.
The translator, solver encoding, structural-to-value argument,
engine/kernel interface, and compiler/GPU execution remain trusted.
Attention assumes input-dependent finiteness, unchecked at runtime.
Full numerical correctness against a reference semantics remains outside our scope.

\section{Evaluation}
\label{sec:eval}

\subsection{Determinism and Numerical Variation}
\label{sec:eval-determinism}

We detail the tests in \autoref{sec:production-measurements},
extend \sys's model coverage, and investigate selected production-engine
mismatches.

\parab{Test coverage.}
Prompts are random token sequences with lengths
chosen near batching, chunking, and attention-window boundaries. Batch tests use
32 prompts of 17--8,193 tokens and batch sizes 2, 4, and 8, yielding 296
comparisons including order variations. Chunk tests
use 18 prompts of 63--32,768 tokens and budgets of 64, 128, 256, 512, and
1,024 tokens against an unchunked reference with a 32,768-token budget, yielding 90
comparisons. Prefill--decode tests start from prompts of 257, 8,192, and
32,768 tokens and compare 128 prediction positions each, yielding 384 pairs.
The 14 prefix-reuse comparisons cover full and partial reuse, divergent
continuations, and generated-prefix reuse. Together, these give 784
comparisons per model and configuration.

Using the criteria of \autoref{sec:production-measurements}, \sys
passes all 784 comparisons for each of seven models: Llama3.1-8B, Llama3.2-3B, Gemma3-4B, Gemma3-12B, Gemma3-27B, Gemma4-12B, and
Gemma4-31B, totaling 5,488 bitwise-matched pairs.

\parab{Metadata-dependent specialization.}
In both vLLM and SGLang, FA3 uses a batch-wide K-length bound to determine
whether local-window masking is needed. This choice also affects kernel
specialization and tiling, so batch composition can change attention outputs. SGLang's decode graphs expose
a related dependence on storage metadata: their wider KV page tables affect
a similar specialization decision, even when the valid keys are unchanged.
These behaviors account for selected batch and prefill--decode differences
in the Gemma3-4B measurements.

\parab{Path-dependent numerical behavior.}
SGLang's Triton prefill and decode paths use different numerical algorithms;
fixing the decode split-KV count does not make their attention outputs bitwise
identical. In vLLM's invariant Triton path, query partitioning changes which
fully masked softmax tiles are visited. Processing such a tile can also replace a
running maximum of $-\infty$ with zero, affecting subsequent rounding. Thus, numerical variation can arise within an attention
implementation as well as from specialization choices.

\parab{Scope.}
Controlled tests support these explanations, which cover some mismatches
in \autoref{tab:production-motivation}. The differences characterize numerical
behavior under tested execution variations and need not indicate engine bugs.

\subsection{Performance}
\label{sec:eval-performance}

\begin{figure*}[t]
  \centering
  \includegraphics[width=\textwidth]{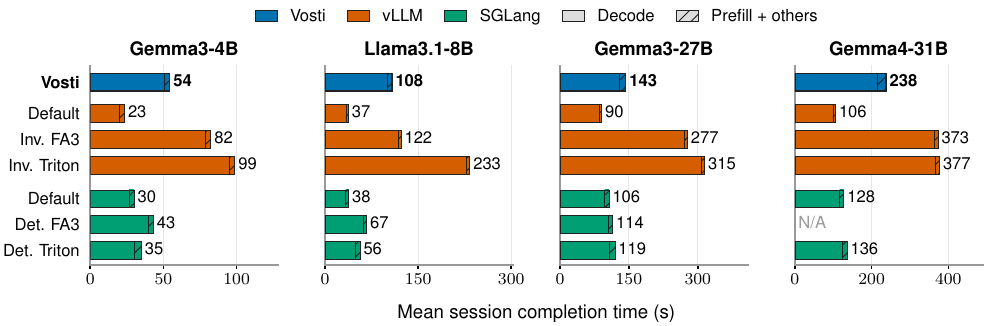}
  \caption{End-to-end session time (lower is better).
  Solid segments show decode time, from first
  output token to response completion; hatched segments combine initial-turn
  and follow-up prefill (request submission to first output token) with
  inter-turn user delay and client/scheduling gaps.
  N/A indicates an unsupported engine configuration.}
  \Description{Four model panels share seven mode rows. The color legend identifies
  \sys in blue, vLLM in orange, and SGLang in green. Each horizontal
  bar shows mean session completion time, with solid decode time followed by
  hatched prefill and other time. Integer totals appear at the bar ends.
  Each model has its own time scale; an unsupported configuration is marked N/A.}
  \label{fig:shared-prefix-sessions}
\end{figure*}

For Gemma4-31B, the vLLM invariant FA3 configuration uses FA3 for local
attention and Triton for global attention because FA3 does not support its
512-dimensional heads. SGLang's deterministic FA3
configuration is unsupported for this model (shown as N/A or omitted).

\parab{Throughput.}
To separate prefill and decode costs, we vary cached context length from 0 to 16K
tokens on a single H200 NVL GPU, holding the batch size at 4
and the per-sequence query/decode length fixed.
\autoref{fig:cached-phases} shows engines' prefill and decode throughput at different cached context lengths. Timing excludes cache warmup.
\sys achieves $1.36$--$3.19\times$ the decode throughput of vLLM's invariant
configurations, while prefill throughput remains lower than that of both engines.

These trends are consistent with the engines' kernel choices.
SGLang benefits from DeepGEMM~\cite{deepgemm2025} and efficient attention,
including FA3 and split-KV in its Triton decode path.
Both \sys and vLLM's invariant modes use Triton matmul,
but vLLM's larger tiles favor prefill, whereas \sys's smaller tiles favor
decode. This trade-off helps explain \sys's lower prefill throughput and
its decode advantage over vLLM's invariant modes despite comparable or
slower attention.

\parab{Session completion time.}
Coding
agents and document-analysis sessions ~\cite{anthropic2025promptcaching, shihipar2026promptcaching} reuse long-lived reference material
or project context across calls, while new messages and tool results extend
the history. With effective cache retention, such agentic workloads
can be decode-dominated~\cite{yuan2026agenticworkloads}.
We capture this pattern with a synthetic LLM-serving trace of four concurrent
sessions of six turns each on an H200 NVL GPU.
Initial-turn prefill processes a 16K-token prompt, reusing an 8K shared
prefix already cached. Follow-up prefill adds 256 input tokens to the
retained conversation. Each turn produces 768 output tokens. Each session issues
its next turn after the preceding response completes and a 0.5\,s inter-turn
user delay. Initial requests are randomly staggered. \autoref{fig:shared-prefix-sessions} reports
mean session completion time, broken down into decode and combined prefill
and other delays.

Decode accounts for 91--93\% of \sys's session time, limiting the
impact of prefill overhead in this decode-heavy workload.
\sys completes sessions $1.14$--$2.20\times$ faster than vLLM's invariant
modes, while remaining slower than the default modes
and SGLang's deterministic modes.

\subsection{Implementation and Verification Effort}
\label{sec:eval-effort}

\autoref{tab:core-effort} reports implementation, specification, and proof
lines of code for \sys's core components, together with verification time. The engine specification comprises two parts. The abstract specification
defines request correspondence, execution traces, and observable output
agreement. Runtime contracts
express the assumptions at external-operation boundaries, including Verus
counterparts of the kernel contracts and assumptions about certain CUDA calls used at runtime.

We measure wall-clock verification time with two AMD EPYC 9355
32-core CPUs (128 hyperthreads). Verus uses 127 workers; the kernel
verifier runs on a single core, with mean times per kernel configuration reported.

\begin{table}[t]
\centering
\footnotesize
\caption{Core implementation and verification effort (lines of code), and
wall-clock verification time of \sys. Kernel times are per configuration.}
\label{tab:core-effort}
\begin{tabular}{lrrrr}
\toprule
\textbf{Component} & \textbf{Impl.} & \textbf{Spec} & \textbf{Proof} &
\textbf{Time} \\
\midrule
Engine and model execution & 14,043 & 876 & 64,952 & 52.9s \\
\midrule
Triton kernels & 3,028 & 500 & 0 & 17.3s \\
\quad Attention (Full \& SWA) & 1,189 & 213 & 0 & 15.8s \\
\quad Linear & 387 & 45 & 0 & 0.4s \\
\quad Normalization & 665 & 99 & 0 & 0.5s \\
\quad Activation \& others & 787 & 143 & 0 & 0.6s \\
\midrule
Kernel verifier & 17,182 & 0 & 0 & - \\
\bottomrule
\end{tabular}
\par\smallskip\noindent\parbox{\linewidth}{\raggedright
\textbf{Engine specification breakdown:} Abstract
specification: 168; runtime contracts: 708 LoC.}
\end{table}

\section{Discussion}
\label{sec:discussion}

\parab{Determinism and functional correctness.}
A consistently incorrect implementation can still be deterministic. As an
extreme example, an engine that always returns all-zero logits and token ID
zero produces identical outputs across executions while failing to implement
the intended model. \sys's verification establishes determinism but not functional
correctness. Our guarantee also fixes the model
and value-relevant deployment configuration. It does not establish bitwise
agreement across numerical dtypes, kernel configurations, compiler versions,
or hardware platforms.

\parab{Cross-kernel relational verification.}
Fixing kernel and launch choices simplifies verification, but a single
implementation may be suboptimal across workload regimes.
DeepSeek-V4 uses two decoding-attention kernels: one assigns a single SM
per sequence for full waves, while the other uses multiple SMs to reduce
latency in partially filled waves. Both preserve the same accumulation
order for bitwise identity~\cite{deepseek2026v4}.
Extending our relational analysis to verify bitwise equivalence across
distinct kernels could enable workload-dependent dispatch while preserving
deterministic outputs.

\parab{Additional model architectures.}
Linear-attention architectures such as Gated DeltaNet summarize token history
in a recurrent state~\cite{yang2025gateddelta}. Supporting them would require extending
our cache invariant to associate retained states with token prefixes and
establishing bitwise agreement across chunked prefill and recurrent decode,
whose mathematically equivalent formulations may group floating-point
operations differently. Mixture-of-experts (MoE) models introduce
data-dependent routing and token regrouping; expert-capacity limits in some
designs can also make computation depend on co-batched
tokens~\cite{fedus2022switch}. Extending \sys to MoE would require
batch-independent routing policies, invariants preserving token identity
through regrouping, and relational contracts for expert computation and
output accumulation.

\parab{Multi-GPU inference.}
\sys currently supports single-GPU execution. Extending determinism to
parallel inference introduces both numerical and distributed-state obligations.
Tensor parallelism requires preserving value-relevant accumulation order
across local kernels and collective reductions, while pipeline execution and
distributed KV storage require tracking request and token provenance across
devices. Our approach could extend through relational contracts for
distributed operations and invariants for communication and state ownership.
An initial target is determinism within a fixed parallel configuration;
bitwise agreement across GPU counts or partitionings requires additional
constraints~\cite{zhang2025tpinvariance}.

\section{Related Work}

\parab{Deterministic LLM inference.}
Thinking Machines Lab introduced batch-invariant kernels, and vLLM, SGLang,
and DeepSeek have brought related modes into production inference
stacks~\cite{he2025defeating, sglang2025deterministic, vllm2025bitwise,
vllm2026batchinvariance, deepseek2026v4}.  Yuan et al. empirically characterize
numerical variation across batch sizes, GPU types, and parallel configurations~\cite{yuan2025numerical},
while Zhang et al. study invariance across tensor-parallel sizes
~\cite{zhang2025tpinvariance}. These efforts control numerical
variation through kernel design and empirical validation. \sys combines engine
invariants with relational kernel contracts to establish per-request determinism
across batching, chunking, prefill--decode transitions, and cache reuse under
its stated assumptions.

LLM-42~\cite{gond2026llm42} uses fixed-shape verification and rollback to
recover deterministic decoding. MarginGate~\cite{chu2026margingate}
selectively verifies low-margin steps and repairs the current KV column,
with thresholds calibrated for sequence-level agreement.
\sys verifies the implementation ahead of execution,
avoiding per-request replay.

\parab{Reproducible floating-point computation.}
The HPC community has developed bitwise-reproducible reductions whose results are
independent of data partitioning or summation order~\cite{demmel2013reproducible}.
\sys targets a fixed deployment, preserving each position's logical inputs
and value-relevant operation sequence across execution variations. This requires
engine invariants alongside kernel-level guarantees.

\parab{KV-cache sharing security.}
HijackKV~\cite{zhang2026hijackkv} shows that position-independent reuse can
carry attacker-controlled context into another request through cached KV
values. This motivates tracking the full causal prefix, as captured by
\sys's cache-provenance invariant.
PromptPeek~\cite{wu2025promptpeek} reconstructs other users' prompts through
serving-order side channels induced by shared prefix caches.
\sys constrains output values across cache reuse; confidentiality of timing
and serving order remains outside its guarantee.

\parab{GPU kernel verification.}
GPUVerify proves race freedom and barrier safety for CUDA and OpenCL
kernels~\cite{betts2012gpuverify}. Volta checks optimized ML kernels against
reference implementations at the PTX level using real-number
arithmetic~\cite{dubey2025volta},
and ProofWright verifies memory, thread, and functional properties of generated
CUDA kernels~\cite{chatterjee2025proofwright}. Product programs reduce relational
verification to verification of a combined program~\cite{barthe2011product}.
\sys relates two invocations
of the same Triton kernel with different packed-row indices, query lengths,
and physical page tables. Its tile-level analysis aligns the selected
computations and establishes equal dependencies. Like translation
validation~\cite{pnueli1998translation} and Alive2~\cite{lopes2021alive2}, it
checks concrete programs' relational obligations; our translator,
cross-language bridge, and compiler remain trusted.

\section{Conclusion}

We presented \sys, an LLM inference system designed and verified against a
system-level specification of determinism. For a fixed model and deployment,
its proof combines engine invariants checked in Verus with relational Triton
kernel contracts. The engine preserves each position's logical inputs across
scheduling and KV-cache reuse, while the kernel verifier establishes equal
selected outputs by aligning value-relevant operations and their dependencies.
Together, they establish per-request bitwise determinism under the stated
assumptions and trusted computing base.

\sys passes all 5,488 bitwise comparisons across seven Llama and Gemma models.
On the evaluated decode-heavy serving workload, it completes sessions
$1.14$--$2.20\times$ faster than vLLM's batch-invariant modes, while
remaining slower than the default configurations and SGLang's deterministic
modes. These results demonstrate how engine invariants and relational kernel
verification can support deterministic inference with practical performance.
\sys's source code is publicly available at \url{https://github.com/QDelta/Vosti}.

\begin{acks}
This work was supported in part by National Science Foundation grants CNS-2625580, CNS-2238665, CNS-2402696, and OAC-2503010, as well as by gifts from Amazon, Google, and Meta. This research was also supported by the
National Science Foundation through the ACCESS program
and AWS through the CloudBank project, which is supported
by National Science Foundation grant 1925001.
\end{acks}

\bibliographystyle{ACM-Reference-Format}
\bibliography{confs_long,ref}


\begin{thebibliography}{37}


\ifx \showCODEN    \undefined \def \showCODEN     #1{\unskip}     \fi
\ifx \showISBNx    \undefined \def \showISBNx     #1{\unskip}     \fi
\ifx \showISBNxiii \undefined \def \showISBNxiii  #1{\unskip}     \fi
\ifx \showISSN     \undefined \def \showISSN      #1{\unskip}     \fi
\ifx \showLCCN     \undefined \def \showLCCN      #1{\unskip}     \fi
\ifx \shownote     \undefined \def \shownote      #1{#1}          \fi
\ifx \showarticletitle \undefined \def \showarticletitle #1{#1}   \fi
\ifx \showURL      \undefined \def \showURL       {\relax}        \fi
\providecommand\bibfield[2]{#2}
\providecommand\bibinfo[2]{#2}
\providecommand\natexlab[1]{#1}
\providecommand\showeprint[2][]{arXiv:#2}

\bibitem[Agrawal et~al\mbox{.}(2024)]%
        {agrawal2024chunkedprefill}
\bibfield{author}{\bibinfo{person}{Amey Agrawal}, \bibinfo{person}{Nitin
  Kedia}, \bibinfo{person}{Ashish Panwar}, \bibinfo{person}{Jayashree Mohan},
  \bibinfo{person}{Nipun Kwatra}, \bibinfo{person}{Bhargav Gulavani},
  \bibinfo{person}{Alexey Tumanov}, {and} \bibinfo{person}{Ramachandran
  Ramjee}.} \bibinfo{year}{2024}\natexlab{}.
\newblock \showarticletitle{{Taming {Throughput-Latency} Tradeoff in {LLM}
  Inference with {Sarathi-Serve}}}. In \bibinfo{booktitle}{\emph{Symposium on
  Operating Systems Design and Implementation (OSDI)}}.
  \bibinfo{pages}{117--134}.
\newblock
\urldef\tempurl%
\url{https://www.usenix.org/conference/osdi24/presentation/agrawal}
\showURL{%
\tempurl}


\bibitem[{Anthropic}(2025)]%
        {anthropic2025promptcaching}
\bibfield{author}{\bibinfo{person}{{Anthropic}}.}
  \bibinfo{year}{2025}\natexlab{}.
\newblock \bibinfo{title}{{Prompt caching with Claude}}.
\newblock \bibinfo{howpublished}{\url{https://claude.com/blog/prompt-caching}}.
\newblock


\bibitem[Barthe et~al\mbox{.}(2011)]%
        {barthe2011product}
\bibfield{author}{\bibinfo{person}{Gilles Barthe}, \bibinfo{person}{Juan~Manuel
  Crespo}, {and} \bibinfo{person}{C{\'e}sar Kunz}.}
  \bibinfo{year}{2011}\natexlab{}.
\newblock \showarticletitle{{Relational Verification Using Product Programs}}.
  In \bibinfo{booktitle}{\emph{International Symposium on Formal Methods
  (FM)}}. \bibinfo{pages}{200--214}.
\newblock
\href{https://doi.org/10.1007/978-3-642-21437-0_17}{doi:\nolinkurl{10.1007/978-3-642-21437-0_17}}


\bibitem[Betts et~al\mbox{.}(2012)]%
        {betts2012gpuverify}
\bibfield{author}{\bibinfo{person}{Adam Betts}, \bibinfo{person}{Nathan Chong},
  \bibinfo{person}{Alastair Donaldson}, \bibinfo{person}{Shaz Qadeer}, {and}
  \bibinfo{person}{Paul Thomson}.} \bibinfo{year}{2012}\natexlab{}.
\newblock \showarticletitle{{GPUVerify: a verifier for GPU kernels}}. In
  \bibinfo{booktitle}{\emph{ACM SIGPLAN Conference on Object-Oriented
  Programming, Systems, Languages, and Applications (OOPSLA)}}.
  \bibinfo{pages}{113--132}.
\newblock
\href{https://doi.org/10.1145/2384616.2384625}{doi:\nolinkurl{10.1145/2384616.2384625}}


\bibitem[Chatterjee et~al\mbox{.}(2025)]%
        {chatterjee2025proofwright}
\bibfield{author}{\bibinfo{person}{Bodhisatwa Chatterjee},
  \bibinfo{person}{Drew Zagieboylo}, \bibinfo{person}{Sana Damani},
  \bibinfo{person}{Siva Hari}, {and} \bibinfo{person}{Christos Kozyrakis}.}
  \bibinfo{year}{2025}\natexlab{}.
\newblock \bibinfo{title}{{ProofWright: Towards Agentic Formal Verification of
  CUDA}}.
\newblock
\showeprint[arxiv]{2511.12294}~[cs.SE]
\urldef\tempurl%
\url{https://arxiv.org/abs/2511.12294}
\showURL{%
\tempurl}


\bibitem[Chen et~al\mbox{.}(2023)]%
        {chen2023specdecoding}
\bibfield{author}{\bibinfo{person}{Charlie Chen}, \bibinfo{person}{Sebastian
  Borgeaud}, \bibinfo{person}{Geoffrey Irving}, \bibinfo{person}{Jean-Baptiste
  Lespiau}, \bibinfo{person}{Laurent Sifre}, {and} \bibinfo{person}{John
  Jumper}.} \bibinfo{year}{2023}\natexlab{}.
\newblock \bibinfo{title}{{Accelerating Large Language Model Decoding with
  Speculative Sampling}}.
\newblock
\showeprint[arxiv]{2302.01318}~[cs.CL]
\urldef\tempurl%
\url{https://arxiv.org/abs/2302.01318}
\showURL{%
\tempurl}


\bibitem[Chu et~al\mbox{.}(2026)]%
        {chu2026margingate}
\bibfield{author}{\bibinfo{person}{Kexin Chu}, \bibinfo{person}{Yang Zhou},
  {and} \bibinfo{person}{Wei Zhang}.} \bibinfo{year}{2026}\natexlab{}.
\newblock \bibinfo{title}{{MarginGate: Sparse Margin-Triggered Verification for
  Batch-Invariant LLM Inference}}.
\newblock
\showeprint[arxiv]{2605.30218}~[cs.LG]
\urldef\tempurl%
\url{https://arxiv.org/abs/2605.30218}
\showURL{%
\tempurl}


\bibitem[Cousot and Cousot(1977)]%
        {cousot1977abstract}
\bibfield{author}{\bibinfo{person}{Patrick Cousot} {and}
  \bibinfo{person}{Radhia Cousot}.} \bibinfo{year}{1977}\natexlab{}.
\newblock \showarticletitle{{Abstract interpretation: a unified lattice model
  for static analysis of programs by construction or approximation of
  fixpoints}}. In \bibinfo{booktitle}{\emph{ACM SIGPLAN-SIGACT Symposium on
  Principles of Programming Languages (POPL)}}. \bibinfo{pages}{238--252}.
\newblock
\href{https://doi.org/10.1145/512950.512973}{doi:\nolinkurl{10.1145/512950.512973}}


\bibitem[Dao et~al\mbox{.}(2022)]%
        {dao2022flashattention}
\bibfield{author}{\bibinfo{person}{Tri Dao}, \bibinfo{person}{Dan Fu},
  \bibinfo{person}{Stefano Ermon}, \bibinfo{person}{Atri Rudra}, {and}
  \bibinfo{person}{Christopher R{\'e}}.} \bibinfo{year}{2022}\natexlab{}.
\newblock \showarticletitle{{FlashAttention: Fast and Memory-Efficient Exact
  Attention with IO-Awareness}}. In \bibinfo{booktitle}{\emph{Conference on
  Neural Information Processing Systems (NeurIPS)}}.
  \bibinfo{pages}{16344--16359}.
\newblock
\href{https://doi.org/10.52202/068431-1189}{doi:\nolinkurl{10.52202/068431-1189}}


\bibitem[de~Moura and Bj{\o}rner(2008)]%
        {demoura2008z3}
\bibfield{author}{\bibinfo{person}{Leonardo de Moura} {and}
  \bibinfo{person}{Nikolaj Bj{\o}rner}.} \bibinfo{year}{2008}\natexlab{}.
\newblock \showarticletitle{{Z3: An Efficient SMT Solver}}. In
  \bibinfo{booktitle}{\emph{International Conference on Tools and Algorithms
  for the Construction and Analysis of Systems (TACAS)}}.
  \bibinfo{pages}{337--340}.
\newblock
\href{https://doi.org/10.1007/978-3-540-78800-3_24}{doi:\nolinkurl{10.1007/978-3-540-78800-3_24}}


\bibitem[{DeepSeek-AI} et~al\mbox{.}(2026)]%
        {deepseek2026v4}
\bibfield{author}{\bibinfo{person}{{DeepSeek-AI}} {et~al\mbox{.}}}
  \bibinfo{year}{2026}\natexlab{}.
\newblock \bibinfo{title}{{DeepSeek-V4: Towards Highly Efficient Million-Token
  Context Intelligence}}.
\newblock
\showeprint[arxiv]{2606.19348}~[cs.CL]
\urldef\tempurl%
\url{https://arxiv.org/abs/2606.19348}
\showURL{%
\tempurl}


\bibitem[Demmel and Nguyen(2013)]%
        {demmel2013reproducible}
\bibfield{author}{\bibinfo{person}{James Demmel} {and}
  \bibinfo{person}{Hong~Diep Nguyen}.} \bibinfo{year}{2013}\natexlab{}.
\newblock \showarticletitle{{Fast Reproducible Floating-Point Summation}}. In
  \bibinfo{booktitle}{\emph{IEEE Symposium on Computer Arithmetic (ARITH)}}.
  \bibinfo{pages}{163--172}.
\newblock
\href{https://doi.org/10.1109/ARITH.2013.9}{doi:\nolinkurl{10.1109/ARITH.2013.9}}


\bibitem[Driscoll et~al\mbox{.}(2026)]%
        {dubey2025volta}
\bibfield{author}{\bibinfo{person}{Benjamin Driscoll}, \bibinfo{person}{Kshitij
  Dubey}, \bibinfo{person}{Anjiang Wei}, \bibinfo{person}{Neeraj Kayal},
  \bibinfo{person}{Rahul Sharma}, {and} \bibinfo{person}{Alex Aiken}.}
  \bibinfo{year}{2026}\natexlab{}.
\newblock \showarticletitle{{Equivalence Checking of ML GPU Kernels}}. In
  \bibinfo{booktitle}{\emph{ACM SIGPLAN Conference on Object-Oriented
  Programming, Systems, Languages, and Applications (OOPSLA)}}.
\newblock
\urldef\tempurl%
\url{https://arxiv.org/abs/2511.12638}
\showURL{%
\tempurl}
\newblock
\shownote{To appear}.


\bibitem[Fedus et~al\mbox{.}(2022)]%
        {fedus2022switch}
\bibfield{author}{\bibinfo{person}{William Fedus}, \bibinfo{person}{Barret
  Zoph}, {and} \bibinfo{person}{Noam Shazeer}.}
  \bibinfo{year}{2022}\natexlab{}.
\newblock \showarticletitle{{Switch Transformers: Scaling to Trillion Parameter
  Models with Simple and Efficient Sparsity}}.
\newblock \bibinfo{journal}{\emph{Journal of Machine Learning Research}}
  \bibinfo{volume}{23}, \bibinfo{number}{120} (\bibinfo{year}{2022}),
  \bibinfo{pages}{1--39}.
\newblock
\urldef\tempurl%
\url{https://jmlr.org/papers/v23/21-0998.html}
\showURL{%
\tempurl}


\bibitem[Gond et~al\mbox{.}(2026)]%
        {gond2026llm42}
\bibfield{author}{\bibinfo{person}{Raja Gond}, \bibinfo{person}{Aditya~K
  Kamath}, \bibinfo{person}{Ramachandran Ramjee}, {and} \bibinfo{person}{Ashish
  Panwar}.} \bibinfo{year}{2026}\natexlab{}.
\newblock \showarticletitle{{LLM-42: Enabling Determinism in LLM Inference with
  Verified Speculation}}. In \bibinfo{booktitle}{\emph{ACM Symposium on
  Operating Systems Principles (SOSP)}}.
\newblock
\urldef\tempurl%
\url{https://www.microsoft.com/en-us/research/publication/llm-42-enabling-determinism-in-llm-inference-with-verified-speculation/}
\showURL{%
\tempurl}
\newblock
\shownote{To appear}.


\bibitem[He and {Thinking Machines Lab}(2025)]%
        {he2025defeating}
\bibfield{author}{\bibinfo{person}{Horace He} {and} \bibinfo{person}{{Thinking
  Machines Lab}}.} \bibinfo{year}{2025}\natexlab{}.
\newblock \bibinfo{title}{{Defeating Nondeterminism in {LLM} Inference}}.
\newblock
  \bibinfo{howpublished}{\url{https://thinkingmachines.ai/blog/defeating-nondeterminism-in-llm-inference/}}.
\newblock
\href{https://doi.org/10.64434/tml.20250910}{doi:\nolinkurl{10.64434/tml.20250910}}


\bibitem[Kwon et~al\mbox{.}(2023)]%
        {kwon2023pagedattention}
\bibfield{author}{\bibinfo{person}{Woosuk Kwon}, \bibinfo{person}{Zhuohan Li},
  \bibinfo{person}{Siyuan Zhuang}, \bibinfo{person}{Ying Sheng},
  \bibinfo{person}{Lianmin Zheng}, \bibinfo{person}{Cody~Hao Yu},
  \bibinfo{person}{Joseph Gonzalez}, \bibinfo{person}{Hao Zhang}, {and}
  \bibinfo{person}{Ion Stoica}.} \bibinfo{year}{2023}\natexlab{}.
\newblock \showarticletitle{{Efficient Memory Management for Large Language
  Model Serving with PagedAttention}}. In \bibinfo{booktitle}{\emph{ACM
  Symposium on Operating Systems Principles (SOSP)}}.
  \bibinfo{pages}{611--626}.
\newblock
\href{https://doi.org/10.1145/3600006.3613165}{doi:\nolinkurl{10.1145/3600006.3613165}}


\bibitem[Lattuada et~al\mbox{.}(2024)]%
        {lattuada2024verus}
\bibfield{author}{\bibinfo{person}{Andrea Lattuada}, \bibinfo{person}{Travis
  Hance}, \bibinfo{person}{Jay Bosamiya}, \bibinfo{person}{Matthias Brun},
  \bibinfo{person}{Chanhee Cho}, \bibinfo{person}{Hayley LeBlanc},
  \bibinfo{person}{Pranav Srinivasan}, \bibinfo{person}{Reto Achermann},
  \bibinfo{person}{Tej Chajed}, \bibinfo{person}{Chris Hawblitzel},
  \bibinfo{person}{Jon Howell}, \bibinfo{person}{Jacob~R. Lorch},
  \bibinfo{person}{Oded Padon}, {and} \bibinfo{person}{Bryan Parno}.}
  \bibinfo{year}{2024}\natexlab{}.
\newblock \showarticletitle{{Verus: A Practical Foundation for Systems
  Verification}}. In \bibinfo{booktitle}{\emph{ACM Symposium on Operating
  Systems Principles (SOSP)}}. \bibinfo{pages}{438--454}.
\newblock
\href{https://doi.org/10.1145/3694715.3695952}{doi:\nolinkurl{10.1145/3694715.3695952}}


\bibitem[Leviathan et~al\mbox{.}(2023)]%
        {leviathan2023specdecoding}
\bibfield{author}{\bibinfo{person}{Yaniv Leviathan}, \bibinfo{person}{Matan
  Kalman}, {and} \bibinfo{person}{Yossi Matias}.}
  \bibinfo{year}{2023}\natexlab{}.
\newblock \showarticletitle{{Fast Inference from Transformers via Speculative
  Decoding}}. In \bibinfo{booktitle}{\emph{International Conference on Machine
  Learning (ICML)}}. \bibinfo{pages}{19274--19286}.
\newblock
\urldef\tempurl%
\url{https://proceedings.mlr.press/v202/leviathan23a.html}
\showURL{%
\tempurl}


\bibitem[Lopes et~al\mbox{.}(2021)]%
        {lopes2021alive2}
\bibfield{author}{\bibinfo{person}{Nuno~P. Lopes}, \bibinfo{person}{Juneyoung
  Lee}, \bibinfo{person}{Chung-Kil Hur}, \bibinfo{person}{Zhengyang Liu}, {and}
  \bibinfo{person}{John Regehr}.} \bibinfo{year}{2021}\natexlab{}.
\newblock \showarticletitle{{Alive2: bounded translation validation for LLVM}}.
  In \bibinfo{booktitle}{\emph{ACM SIGPLAN Conference on Programming Language
  Design and Implementation (PLDI)}}. \bibinfo{pages}{65--79}.
\newblock
\href{https://doi.org/10.1145/3453483.3454030}{doi:\nolinkurl{10.1145/3453483.3454030}}


\bibitem[Pnueli et~al\mbox{.}(1998)]%
        {pnueli1998translation}
\bibfield{author}{\bibinfo{person}{Amir Pnueli}, \bibinfo{person}{Michael
  Siegel}, {and} \bibinfo{person}{Eli Singerman}.}
  \bibinfo{year}{1998}\natexlab{}.
\newblock \showarticletitle{{Translation Validation}}. In
  \bibinfo{booktitle}{\emph{International Conference on Tools and Algorithms
  for the Construction and Analysis of Systems (TACAS)}}.
  \bibinfo{pages}{151--166}.
\newblock
\href{https://doi.org/10.1007/BFb0054170}{doi:\nolinkurl{10.1007/BFb0054170}}


\bibitem[Shah et~al\mbox{.}(2024)]%
        {shah2024flashattention3}
\bibfield{author}{\bibinfo{person}{Jay Shah}, \bibinfo{person}{Ganesh
  Bikshandi}, \bibinfo{person}{Ying Zhang}, \bibinfo{person}{Vijay Thakkar},
  \bibinfo{person}{Pradeep Ramani}, {and} \bibinfo{person}{Tri Dao}.}
  \bibinfo{year}{2024}\natexlab{}.
\newblock \showarticletitle{{FlashAttention-3: Fast and Accurate Attention with
  Asynchrony and Low-precision}}. In \bibinfo{booktitle}{\emph{Conference on
  Neural Information Processing Systems (NeurIPS)}}.
  \bibinfo{pages}{68658--68685}.
\newblock
\href{https://doi.org/10.52202/079017-2193}{doi:\nolinkurl{10.52202/079017-2193}}


\bibitem[Shihipar(2026)]%
        {shihipar2026promptcaching}
\bibfield{author}{\bibinfo{person}{Thariq Shihipar}.}
  \bibinfo{year}{2026}\natexlab{}.
\newblock \bibinfo{title}{{Lessons from building Claude Code: Prompt caching is
  everything}}.
\newblock
  \bibinfo{howpublished}{\url{https://claude.com/blog/lessons-from-building-claude-code-prompt-caching-is-everything}}.
\newblock


\bibitem[{The SGLang Team}(2025)]%
        {sglang2025deterministic}
\bibfield{author}{\bibinfo{person}{{The SGLang Team}}.}
  \bibinfo{year}{2025}\natexlab{}.
\newblock \bibinfo{title}{{Towards Deterministic Inference in {SGLang} and
  Reproducible {RL} Training}}.
\newblock
  \bibinfo{howpublished}{\url{https://www.lmsys.org/blog/2025-09-22-sglang-deterministic/}}.
\newblock


\bibitem[{The vLLM Team}(2026)]%
        {vllm2026batchinvariance}
\bibfield{author}{\bibinfo{person}{{The vLLM Team}}.}
  \bibinfo{year}{2026}\natexlab{}.
\newblock \bibinfo{title}{{Batch Invariance}}.
\newblock
  \bibinfo{howpublished}{\url{https://docs.vllm.ai/en/latest/features/batch_invariance/}}.
\newblock
\newblock
\shownote{Accessed on 2026-09-08}.


\bibitem[Tillet et~al\mbox{.}(2019)]%
        {tillet2019triton}
\bibfield{author}{\bibinfo{person}{Philippe Tillet}, \bibinfo{person}{H.~T.
  Kung}, {and} \bibinfo{person}{David Cox}.} \bibinfo{year}{2019}\natexlab{}.
\newblock \showarticletitle{{Triton: an intermediate language and compiler for
  tiled neural network computations}}. In \bibinfo{booktitle}{\emph{ACM SIGPLAN
  International Workshop on Machine Learning and Programming Languages
  (MAPL)}}. \bibinfo{pages}{10--19}.
\newblock
\href{https://doi.org/10.1145/3315508.3329973}{doi:\nolinkurl{10.1145/3315508.3329973}}


\bibitem[{vLLM and TorchTitan Teams}(2025)]%
        {vllm2025bitwise}
\bibfield{author}{\bibinfo{person}{{vLLM and TorchTitan Teams}}.}
  \bibinfo{year}{2025}\natexlab{}.
\newblock \bibinfo{title}{{No More Train-Inference Mismatch: Bitwise Consistent
  On-Policy Reinforcement Learning with vLLM and TorchTitan}}.
\newblock
  \bibinfo{howpublished}{\url{https://vllm.ai/blog/2025-11-10-bitwise-consistent-train-inference}}.
\newblock


\bibitem[Wu et~al\mbox{.}(2025)]%
        {wu2025promptpeek}
\bibfield{author}{\bibinfo{person}{Guanlong Wu}, \bibinfo{person}{Zheng Zhang},
  \bibinfo{person}{Yao Zhang}, \bibinfo{person}{Weili Wang},
  \bibinfo{person}{Jianyu Niu}, \bibinfo{person}{Ye Wu}, {and}
  \bibinfo{person}{Yinqian Zhang}.} \bibinfo{year}{2025}\natexlab{}.
\newblock \showarticletitle{{I Know What You Asked: Prompt Leakage via
  {KV-Cache} Sharing in {Multi-Tenant LLM} Serving}}. In
  \bibinfo{booktitle}{\emph{Network and Distributed System Security Symposium
  (NDSS)}}.
\newblock
\href{https://doi.org/10.14722/ndss.2025.241772}{doi:\nolinkurl{10.14722/ndss.2025.241772}}


\bibitem[Yang et~al\mbox{.}(2025)]%
        {yang2025gateddelta}
\bibfield{author}{\bibinfo{person}{Songlin Yang}, \bibinfo{person}{Jan Kautz},
  {and} \bibinfo{person}{Ali Hatamizadeh}.} \bibinfo{year}{2025}\natexlab{}.
\newblock \showarticletitle{{Gated Delta Networks: Improving Mamba2 with Delta
  Rule}}. In \bibinfo{booktitle}{\emph{International Conference on Learning
  Representations (ICLR)}}.
\newblock
\urldef\tempurl%
\url{https://proceedings.iclr.cc/paper_files/paper/2025/hash/4904fad153f6434a7bcf04465d4be2cc-Abstract-Conference.html}
\showURL{%
\tempurl}


\bibitem[Yu et~al\mbox{.}(2022)]%
        {yu2022orca}
\bibfield{author}{\bibinfo{person}{Gyeong-In Yu}, \bibinfo{person}{Joo~Seong
  Jeong}, \bibinfo{person}{Geon-Woo Kim}, \bibinfo{person}{Soojeong Kim}, {and}
  \bibinfo{person}{Byung-Gon Chun}.} \bibinfo{year}{2022}\natexlab{}.
\newblock \showarticletitle{{Orca: A Distributed Serving System for
  {Transformer-Based} Generative Models}}. In
  \bibinfo{booktitle}{\emph{Symposium on Operating Systems Design and
  Implementation (OSDI)}}. \bibinfo{pages}{521--538}.
\newblock
\urldef\tempurl%
\url{https://www.usenix.org/conference/osdi22/presentation/yu}
\showURL{%
\tempurl}


\bibitem[Yuan et~al\mbox{.}(2025)]%
        {yuan2025numerical}
\bibfield{author}{\bibinfo{person}{Jiayi Yuan}, \bibinfo{person}{Hao Li},
  \bibinfo{person}{Xinheng Ding}, \bibinfo{person}{Wenya Xie},
  \bibinfo{person}{Yu-Jhe Li}, \bibinfo{person}{Wentian Zhao},
  \bibinfo{person}{Kun Wan}, \bibinfo{person}{Jing Shi}, \bibinfo{person}{Xia
  Hu}, {and} \bibinfo{person}{Zirui Liu}.} \bibinfo{year}{2025}\natexlab{}.
\newblock \showarticletitle{{Understanding and Mitigating Numerical Sources of
  Nondeterminism in LLM Inference}}. In \bibinfo{booktitle}{\emph{Conference on
  Neural Information Processing Systems (NeurIPS)}}.
  \bibinfo{pages}{169819--169851}.
\newblock
\href{https://doi.org/10.52202/085713-5653}{doi:\nolinkurl{10.52202/085713-5653}}


\bibitem[Yuan et~al\mbox{.}(2026)]%
        {yuan2026agenticworkloads}
\bibfield{author}{\bibinfo{person}{Yichao Yuan}, \bibinfo{person}{Ankita
  Nayak}, \bibinfo{person}{Souvik Kundu}, {and} \bibinfo{person}{Nishil
  Talati}.} \bibinfo{year}{2026}\natexlab{}.
\newblock \bibinfo{title}{{Agentic AI Workload Characteristics}}.
\newblock
\showeprint[arxiv]{2605.26297}~[cs.DC]
\urldef\tempurl%
\url{https://arxiv.org/abs/2605.26297}
\showURL{%
\tempurl}


\bibitem[Zhang et~al\mbox{.}(2026a)]%
        {zhang2026batchspec}
\bibfield{author}{\bibinfo{person}{Ranran~Haoran Zhang},
  \bibinfo{person}{Soumik Dey}, \bibinfo{person}{Ashirbad Mishra},
  \bibinfo{person}{Hansi Wu}, \bibinfo{person}{Binbin Li}, {and}
  \bibinfo{person}{Rui Zhang}.} \bibinfo{year}{2026}\natexlab{a}.
\newblock \showarticletitle{{Correctness Forensics for Batch Speculative
  Decoding: Diagnosing the Ragged Tensor Problem}}. In
  \bibinfo{booktitle}{\emph{Findings of the Association for Computational
  Linguistics: EMNLP}}.
\newblock
\urldef\tempurl%
\url{https://arxiv.org/abs/2510.22876}
\showURL{%
\tempurl}
\newblock
\shownote{To appear}.


\bibitem[Zhang et~al\mbox{.}(2026c)]%
        {zhang2026hijackkv}
\bibfield{author}{\bibinfo{person}{Yichi Zhang}, \bibinfo{person}{Zhiqi Wang},
  \bibinfo{person}{Huan Zhang}, {and} \bibinfo{person}{Yuchen Yang}.}
  \bibinfo{year}{2026}\natexlab{c}.
\newblock \showarticletitle{{{HijackKV}: New Threat in {Position-Independent}
  {KV} Cache Reuse}}. In \bibinfo{booktitle}{\emph{USENIX Security Symposium}}.
\newblock
\urldef\tempurl%
\url{https://www.usenix.org/conference/usenixsecurity26/presentation/zhang-yichi}
\showURL{%
\tempurl}


\bibitem[Zhang et~al\mbox{.}(2026b)]%
        {zhang2025tpinvariance}
\bibfield{author}{\bibinfo{person}{Ziyang Zhang}, \bibinfo{person}{Xinheng
  Ding}, \bibinfo{person}{Jiayi Yuan}, \bibinfo{person}{Rixin Liu},
  \bibinfo{person}{Huizi Mao}, \bibinfo{person}{Jiarong Xing}, {and}
  \bibinfo{person}{Zirui Liu}.} \bibinfo{year}{2026}\natexlab{b}.
\newblock \showarticletitle{{Deterministic Inference across Tensor Parallel
  Sizes That Eliminates Training-Inference Mismatch}}. In
  \bibinfo{booktitle}{\emph{International Conference on Machine Learning
  (ICML)}}.
\newblock
\urldef\tempurl%
\url{https://openreview.net/forum?id=5eZmlUyFpl}
\showURL{%
\tempurl}


\bibitem[Zhao et~al\mbox{.}(2025)]%
        {deepgemm2025}
\bibfield{author}{\bibinfo{person}{Chenggang Zhao}, \bibinfo{person}{Zhean Xu},
  \bibinfo{person}{Liang Zhao}, \bibinfo{person}{Jiashi Li},
  \bibinfo{person}{Chenhao Xu}, \bibinfo{person}{Anyi Xu},
  \bibinfo{person}{Shengyu Liu}, \bibinfo{person}{Kexing Zhou}, {and}
  \bibinfo{person}{Kuai Yu}.} \bibinfo{year}{2025}\natexlab{}.
\newblock \bibinfo{title}{{{DeepGEMM}: Clean and Efficient {BLAS} Kernel
  Library on {GPU}}}.
\newblock
  \bibinfo{howpublished}{\url{https://github.com/deepseek-ai/DeepGEMM}}.
\newblock


\bibitem[Zheng et~al\mbox{.}(2024)]%
        {zheng2024sglang}
\bibfield{author}{\bibinfo{person}{Lianmin Zheng}, \bibinfo{person}{Liangsheng
  Yin}, \bibinfo{person}{Zhiqiang Xie}, \bibinfo{person}{Chuyue Sun},
  \bibinfo{person}{Jeff Huang}, \bibinfo{person}{Cody~Hao Yu},
  \bibinfo{person}{Shiyi Cao}, \bibinfo{person}{Christos Kozyrakis},
  \bibinfo{person}{Ion Stoica}, \bibinfo{person}{Joseph~E. Gonzalez},
  \bibinfo{person}{Clark Barrett}, {and} \bibinfo{person}{Ying Sheng}.}
  \bibinfo{year}{2024}\natexlab{}.
\newblock \showarticletitle{{SGLang: Efficient Execution of Structured Language
  Model Programs}}. In \bibinfo{booktitle}{\emph{Conference on Neural
  Information Processing Systems (NeurIPS)}}. \bibinfo{pages}{62557--62583}.
\newblock
\href{https://doi.org/10.52202/079017-2000}{doi:\nolinkurl{10.52202/079017-2000}}


\end{thebibliography}

\appendix
\section{Verifying Paged Attention}
\label{app:attention-verification}

We expand \autoref{sec:pagedattention} using a mathematical presentation of
the attention tile IR.

\subsection{Properties and Logical Inputs}

Fix the element types, head geometry, scale, page size $P$, and tile sizes
$B_M,B_N$, with $B_N\mid P$. Consider one query head of dimension $d$ and
its KV head. Invocation $X$ has query/key lengths $0<s_X\le n_X$;
query row $0\le i_X<s_X$ represents position $p_X=n_X-s_X+i_X$.

Let $Q_X,O_X$ denote queries and outputs. Logical cache matrices $K_X,V_X$
use page table $\pi_X$: for key position $j$,
\[
K_X[j,:]=K_X^{\mathrm{phys}}
  [\pi_X(\lfloor j/P\rfloor),j\bmod P,:],
\]
and similarly for $V_X$. The engine supplies valid lengths and layouts;
physical addresses may differ.

\emph{Batch invariance} relates a packed request to its singleton run
with equal queries and logical cache. A second contract requires equal output
rows at the same logical token position across singleton runs with
different query/key lengths and row indices.
Writing $\equiv_b$ for bitwise equality, this contract requires
\[
\begin{gathered}
p_L=p_R=p,\qquad Q_L[i_L,:]\equiv_b Q_R[i_R,:],\\
K_L[0:p+1,:]\equiv_b K_R[0:p+1,:],\\
V_L[0:p+1,:]\equiv_b V_R[0:p+1,:]\\[2pt]
\Longrightarrow\quad O_L[i_L,:]\equiv_b O_R[i_R,:].
\end{gathered}
\]
This implication uses the assumptions below; concrete annotations
conservatively require equality of whole pages intersecting the prefix. The two
properties together relate packed prefill to decode.

\subsection{A Simplified Tile IR}

Omitting $X$, the selected query tile starts at $a=B_M\lfloor i/B_M\rfloor$.
Its $B_M\times d$ query matrix $Q_t$ processes key tiles $t=0,\ldots,N-1$,
where
\[
N=\left\lceil
  \frac{n-s+\min(s,a+B_M)}{B_N}
\right\rceil.
\]
Page lookups supply $B_N\times d$ tiles $K_t,V_t$, with zero-padded
out-of-range loads. For lanes $0\le r<B_M$, $0\le b<B_N$, and
$j=tB_N+b$, the mask is
\[
M_t[r,b]=(a+r<s)\land(j<n)\land(j\le n-s+a+r).
\]

Following FlashAttention's tiled online-softmax computation~\cite{dao2022flashattention},
the FP32 state holds row maxima $m$, normalization sums $z$, and weighted values
$U$, of shapes $B_M$, $B_M$, and $B_M\times d$, initialized to
$(-\infty,0,0)$. With shared base-two scale $c$, each iteration proceeds as follows.
Operations retain fixed floating-point types, parameters, and value-relevant
order; reductions run over key lanes and row scalars broadcast over columns.
\[
\begin{aligned}
G &\gets \operatorname{dot}(Q_t,K_t^{\mathsf T})\,c,\\
\widehat G &\gets \operatorname{select}(M_t,G,-\infty),\\
h &\gets \operatorname{any}_b M_t,\\
\widehat m &\gets \max(m,\operatorname{max}_b\widehat G),\\
\alpha &\gets \operatorname{exp}_2(m-\widehat m),\\
A &\gets \operatorname{select}
  (M_t,\operatorname{exp}_2(G-\widehat m),0),\\
\widehat z &\gets z\alpha+\operatorname{sum}_b A,\\
\widehat U &\gets \operatorname{dot}(\operatorname{cast}_{V}(A),V_t,U\alpha),\\
(m,z,U)&\gets\operatorname{select}
  (h,(\widehat m,\widehat z,\widehat U),(m,z,U)).
\end{aligned}
\]
Here $\operatorname{dot}$ is the fixed tile product, its optional third
operand is an accumulator, and $\operatorname{select}$ chooses componentwise.
The weights are cast to the cache element type before the second dot;
the normalization sum uses the uncast weights. The stored output is
\[
O_t=\operatorname{cast}_{O}\!\left(
U\operatorname{select}(z>0,1/z,0)\right).
\]
Here $\operatorname{cast}_{O}$ converts to the output element type.
We omit the auxiliary log-sum-exp output from this selected-row contract.

\subsection{Forward Value Analysis}

The analysis tracks index predicates for known zero, one, negative-infinity,
and Boolean values, plus provenance identifying unchanged incoming state.
Unclassified cells remain unknown. Masking and selection give
\[
\begin{aligned}
\neg M_t[r,b]&\Longrightarrow\widehat G[r,b]=-\infty,\\
\neg M_t[r,b]&\Longrightarrow A[r,b]=0,\\
(\forall b,\neg M_t[r,b])&\Longrightarrow h[r]=\mathsf{false}.
\end{aligned}
\]
The false selection arm establishes zero weights, preserved by the cast.
When $h[r]$ is false, the final selection copies the incoming state bit for bit, including
when unused candidate intermediates contain exceptional values.

\subsection{Backward Regional Dependencies}

Backward analysis starts from the selected output's $U$ and $z$ rows.
For $r_X=i_X-a_X$, define the visible key lanes
\[
J_t^X=\{b\mid 0\le b<B_N\ \land\ M_t^X[r_X,b]\}.
\]
Omitting $X$ again, the masked maximum and weights require scores only on
$J_t$. The dot rule
maps these dependencies to $Q_t[r,:]$ and $K_t[J_t,:]$, preserving the full
head-dimension reduction and its order.

Zero weights outside $J_t$ restrict dependencies to $V_t[J_t,:]$.
This assumes finite entries in the loaded $V_t$, including masked lanes,
and the trusted masked-dot rule's bit-level treatment of zero products and
accumulation. Finiteness is not checked at runtime. The rule removes
equality obligations on masked values while retaining every executed lane.

The verifier checks the raw query--key product and state update separately,
aligning operators, operand order, and shared scalars. Z3 checks region and
scalar obligations. Equal incoming state, selected queries, visible KV,
and masks then give equal selected next state.

\subsection{Relational Alignment and Composition}

For full attention with corresponding queries at $p$, both runs visit the first
$C=\lceil(p+1)/B_N\rceil$ tiles in order. Their selected lanes $r_L,r_R$
satisfy, for $0\le t<C$ and $0\le b<B_N$,
\[
M_t^L[r_L,b]=M_t^R[r_R,b]=(tB_N+b\le p).
\]
Since $p<n_L,n_R$, both masks agree despite different lengths and padding.
Page mappings supply the required logical KV equalities.

Extra tiles for later queries have $t\ge C$ and all-false selected masks,
hence preserve selected state exactly. This gives control-flow correspondence
up to identity steps: induction over paired updates establishes equal state,
extra iterations preserve it, and common normalization gives equal outputs.

Sliding-window attention adds $j>p-w$ for shared window size $w>0$.
The implemented traversal starts at
\[
L_X=\left\lfloor
\frac{\max(0,n_X-s_X+a_X-w+1)}{B_N}
\right\rfloor
\]
and ends at $N_X-1$. Both runs include all tiles intersecting the selected
window, from $\lfloor\max(0,p-w+1)/B_N\rfloor$ through $C-1$, in order.
Any extra leading or trailing tiles have all-false selected masks and
preserve state exactly. Thus the same alignment argument applies even when
the loop starts differ. The contract retains full-prefix equality and
supplies no cache-eviction guarantee.

Z3 checks the region, mask, and loop-geometry obligations. Lifting these
checks to bitwise value equality and composing the iteration argument rely
on the trusted verifier reasoning and operation rules in
\autoref{sec:pagedattention}, together with the compiler/GPU assumptions.
Tests support those assumptions on our configurations; input-dependent
finiteness remains an explicit premise.

\end{document}